\documentclass[10pt]{article}
\usepackage{amssymb,amsfonts,amsmath,amsthm,bm}
\usepackage{framed}
\usepackage{graphicx}
\usepackage{psfrag}
\usepackage{subfigure}
\usepackage{algorithm}
\usepackage{hyperref}
\usepackage{enumerate}
\usepackage{color}
\usepackage{palatino}
\usepackage{algcompatible}
\usepackage{booktabs}
\usepackage{multirow}

\long\def\symbolfootnote[#1]#2{\begingroup%
\def\thefootnote{\fnsymbol{footnote}}\footnote[#1]{#2}\endgroup}

\newcommand{\Expect}[1]{\mbox{}{\mathbb{E}}\left[#1\right]}

\newcommand{\FNorm }[1]{\mbox{}\|#1\|_\mathrm{F}  }
\newcommand{\FNormS}[1]{\mbox{}\|#1\|_\mathrm{F}^2}

\newcommand{\TNorm }[1]{\mbox{}\|#1\|_2  }
\newcommand{\TNormS}[1]{\mbox{}\|#1\|_2^2}

\newcommand{\XNorm }[1]{\mbox{}\|#1\|_{\xi}  }
\newcommand{\XNormS}[1]{\mbox{}\|#1\|_{\xi}^2}

\newtheorem{theorem}{\bf Theorem}
\newtheorem{lemma}[theorem]{Lemma}
\newtheorem{definition}[theorem]{Definition}
\newtheorem{proposition}[theorem]{Proposition}

\newcommand{\transp}{^{\textsc{T}}}

\newcommand{\mat}[1]{{\ensuremath{\bm{\mathrm{#1}}}}}

\newcommand{\pinv}[1]{ {#1}^{\dagger}}

\def\rank{\hbox{\rm rank}}

\def\b{{\mathbf b}}
\def\e{{\mathbf e}}

\def\A{\matA}

\def\matA{\mat{A}}
\def\matB{\mat{B}}
\def\matC{\mat{C}}

\def\matE{\mat{E}}

\def\matI{\mat{I}}

\def\matQ{\mat{Q}}
\def\matR{\mat{R}}
\def\matS{\mat{S}}

\def\matU{\mat{U}}
\def\matV{\mat{V}}
\def\matW{\mat{W}}
\def\matX{\mat{X}}
\def\matY{\mat{Y}}
\def\matZ{\mat{Z}}
\def\matSig{\mat{\Sigma}}

\def\matPsi{\mat{\Psi}}

\DeclareMathSymbol{\Prob}{\mathbin}{AMSb}{"50}
\newcommand\remove[1]{}
\newcommand\ignore[1]{}

\def\math#1{$#1$}

\def\frac#1#2{{#1\over #2}}

\def\eqan#1{\begin{eqnarray*}
#1
\end{eqnarray*}}

\DeclareMathSymbol{\R}{\mathbin}{AMSb}{"52}

\newcommand{\argmin}{\operatorname*{argmin}}

\def\x{{\mathbf x}}

\def\a{{\mathbf a}}
\def\b{{\mathbf b}}

\def\dotfil{\leaders\hbox to 1.5mm{.}\hfill}
\newcounter{rmnum}
\def\RN#1{\setcounter{rmnum}{#1}\uppercase\expandafter{\romannumeral\value{rmnum}}}
\def\rn#1{\setcounter{rmnum}{#1}\expandafter{\romannumeral\value{rmnum}}}

\begin{document}

\title{ {\bf Provable Deterministic Leverage Score Sampling \footnote{An extended abstract of this article appeared in the 20th ACM SIGKDD Conference on Knowledge Discovery and Data Mining.}}}
\author{
Dimitris Papailiopoulos\thanks{Electrical and Computer Engineering, UT Austin, dimitris@utexas.edu.
}
\and
Anastasios Kyrillidis\thanks{
School of Computer and Communication Sciences, EPFL, anastasios.kyrillidis@epfl.ch}
\and
Christos Boutsidis\thanks{
Yahoo! Labs, New York, boutsidis@yahoo-inc.com}
}

\date{}
\maketitle

\begin{abstract}
We explain \emph{theoretically} a curious empirical phenomenon:
``Approximating a matrix by {\it deterministically} selecting a subset of its columns with the corresponding 
largest leverage scores results in a good low-rank matrix surrogate''. 
To obtain provable guarantees, previous work requires randomized sampling
of the columns with probabilities proportional to their leverage scores.

In this work, we provide 
a novel theoretical analysis of \emph{deterministic leverage score sampling}. 
We show that such deterministic sampling can be provably as accurate as its randomized counterparts, 
if the leverage scores follow a moderately steep power-law decay.
We support this power-law assumption by 
providing empirical evidence that such decay laws are abundant in real-world data sets.
We then demonstrate empirically the performance of deterministic leverage score sampling, 
which many times matches or outperforms the state-of-the-art techniques.
\end{abstract}

\section{Introduction}\label{sec:intro}

Recently, there has been a lot of interest on selecting the ``best'' or ``more representative'' columns from a data matrix~\cite{DMM08,DM09}.
Qualitatively, these columns reveal the most important information hidden in the underlying matrix structure.
This is similar to what
principal components carry, as extracted via Principal Components Analysis (PCA)~\cite{Jol02}.
In sharp contrast to PCA, using actual columns of the data matrix to form a low-rank surrogate offers interpretability, making it more attractive to practitioners and data analysts~\cite{sun2007less,BMD08,tong2008colibri,DM09}.

To make the discussion precise and to rigorously characterize the ``best'' columns of a matrix,
let us introduce the following  \emph{Column Subset Selection Problem}  (CSSP). 

\medskip
\textsc{Column Subset Selection Problem.} \textit{Let $\matA \in \R^{m \times n}$ and let
$c < n$ be a sampling parameter. Find $c$ columns of $\matA$ -- denoted as
$\matC \in \R^{m \times c}$ -- that minimize 
$$\FNorm{\matA - \matC \pinv{\matC}\matA} \;\;\text{or}\;\; \TNorm{\matA - \matC \pinv{\matC}\matA},$$
where $\pinv{\matC}$ denotes the Moore-Penrose pseudo-inverse.}

\medskip
State of the art algorithms for the CSSP utilize both deterministic and randomized techniques; we discuss related work in Section~\ref{sec:related}.
Here, we describe two algorithms from prior literature that suffice to highlight our contributions.

A central part of our discussion will involve the {\it leverage scores} of a matrix  $\matA$, which we define below.

\begin{definition}\label{def:lev}
[Leverage scores] 
Let $\matV_k \in \mathbb{R}^{n \times k}$ contain the top $k$ right singular vectors of a $m \times n$ matrix 
$\matA$ with rank $\rho = \rank(\matA) \ge k$. 
Then,
the (rank-$k$) leverage score of the $i$-th column of $\matA$ is defined as 
$$
\ell_i^{(k)} = \TNormS{ [\matV_k]_{i,:} }, \quad i = 1, 2, \dots, n.
$$
Here, $[\matV_k]_{i,:}$ denotes the $i$-th row of $\matV_k$.
\end{definition}

One of the first algorithms for column subset selection dates back to 1972: in~\cite{Jol72}, Joliffe
proposes a deterministic sampling of
 the columns of $\matA$ that correspond to the largest  leverage scores $\ell_i^{(k)}$, for some $k < \rank(\matA)$. 
Although this simple approach has been extremely successful in practice~\cite{Jol72, Jol73,paschou2007pca,broadbent2010subset}, to the best of our knowledge, there has been no theoretical explanation why the approximation errors $\FNorm{\matA - \matC \pinv{\matC}\matA}$ and $\TNorm{\matA - \matC \pinv{\matC}\matA}$ should be small.

One way to circumvent the lack of a theoretical analysis for the above deterministic algorithm is by utilizing randomization.
Drineas et al.~\cite{DMM08} proposed the following approach: 
for a target rank $k < \rank(\matA)$, define a probability distribution over the columns of $\matA$, {\it i.e.,} the $i$th column is associated with a probability
$$ p_i = \ell_i^{(k)} / k;$$
observe that $\sum_{i} p_i=1,$ since $\sum_{i} \ell_i^{(k)} = \FNormS{\matV_k} = k$.
Then, in $c$ independent and identically distributed passes, sample with replacement $c$ columns from $\matA$, with probabilities given by $p_i$.
Drineas et al.~\cite{DMM08}, using results in~\cite{RV07}, show that this random subset of columns
$\matC \in \R^{m \times c}$ approximates $\matA$, with constant probability, within relative error: 
$ \FNorm{\matA - \matC \pinv{\matC}\matA} \le  \left( 1 + \varepsilon \right) \FNorm{\matA - \matA_k},$
when the number of sampled columns is 
$c = O(k \log k / \varepsilon^2),$ for some $0 < \varepsilon < 1$.
Here, $\matA_k \in \R^{m \times n}$ is the best rank-$k$  matrix obtained via the SVD.

There are two important remarks that need to be made: {\it (i)}
the randomized algorithm in~\cite{DMM08}  yields a matrix estimate that is ``near optimal'', 
i.e., has error close to that of the  best rank-$k$ approximation;  and {\it (ii)}
the above random sampling algorithm is a straightforward randomized version of the
deterministic algorithm of Joliffe~\cite{Jol72}.

From a practical perspective, the deterministic algorithm of Joliffe~\cite{Jol72} is extremely simple to implement, and is computationally efficient. Unfortunately, as of now, it did not admit provable performance guarantees.
An important open question~\cite{DMM08,paschou2007pca,broadbent2010subset} is:
\emph{Can one simply keep the columns having the largest
leverage scores, as suggested in~\cite{Jol72}, and still have a provably tight approximation?}

\subsection{Contributions}
In this work, we establish a new theoretical analysis for the deterministic leverage score sampling algorithm of Joliffe~\cite{Jol72}.
We show that if the leverage scores $\ell_i^{(k)}$ follow a sufficiently steep power-law decay, then this deterministic algorithm has provably similar or better performance to its randomized counterparts~(see Theorems~\ref{thm1} and \ref{thm2} in Section~\ref{sec:main}).
This means that under the power-law decay assumption, deterministic leverage score sampling provably obtains near optimal low-rank approximations and it can be as accurate as the ``best'' algorithms in the literature~\cite{BDM11a,GS12}.

From an applications point of view,
we support the power law decay assumption of our theoretical analysis by demonstrating  that several real-world data-sets have leverage scores following such decays.
We further run several experiments on synthetic and real data sets, and compare deterministic leverage score sampling with the state of the art algorithms for the CSSP.
In most experiments, the deterministic algorithm obtains tight low-rank approximations, and is shown to perform similar, if not better, than the state of the art.

\subsection{Notation}\label{sec:pre}

We use \math{\matA,\matB,\ldots} to denote matrices and
\math{\a,\b,\ldots} to denote column vectors.
$\matI_{n}$ is the $n \times n$
identity matrix;  $\bm{0}_{m \times n}$ is the $m \times n$ matrix of zeros;
$\e_i$ belongs to the standard basis (whose dimensionality will be clear from the context).
Let 
$$\matC=[\a_{i_1},\ldots,\a_{i_c}] \in \R^{m \times c},
$$
contain \math{c}
columns of~\math{\matA}.
We can equivalently write
\math{\matC=\matA \matS}, where the \emph{sampling matrix} is
\math{\matS=[\e_{i_1},\ldots,\e_{i_c}] \in \R^{n \times c}}.
 \label{chap23}
We define the Frobenius and the spectral norm of a matrix as
$ \FNormS{\matA} = \sum_{i,j} \matA_{ij}^2$ and
$\TNorm{\matA} = \max_{\x:\TNorm{\x}=1}\TNorm{\matA \x}$, respectively.

\section{Deterministic Column Sampling}\label{sec:main}

In this section, we describe the details of the deterministic leverage score sampling algorithm.
In Section~\ref{sec:approx}, we state our approximation guarantees.
In the remaining of the text, given a matrix $\matA$ of rank $\rho,$
we assume that the ``target rank'' is $k < \rho$.
This means that we wish to approximate $\matA$
using a subset of $c \ge k$ of its columns, such that the resulting matrix has an error close to that of the best rank-$k$ approximation.

The deterministic leverage score sampling algorithm can be summarized in the following three steps:

\medskip
\noindent {\bf Step 1:} Obtain ${\bf V}_k$, the top-$k$ right singular vectors of $\matA$.
This can be carried by simply computing the singular value decomposition (SVD) of $\matA$ in
$O(\min\{m,n\} m n)$ time.

\medskip
\noindent {\bf Step 2:} Calculate the leverage scores $\ell_i^{(k)}$.
For simplicity, we assume that $\ell_i^{(k)}$ are sorted in descending order, hence the columns of ${\bf A}$ have the same sorting as well.\footnote{Otherwise, one needs to sort them in $O(n \log n)$ time-cost.}

\medskip
\noindent {\bf Step 3:} Output the $c$ columns of $\matA$ that correspond to the largest $c$ leverage scores $\ell_i^{(k)}$ such that their sum $\sum_{i=1}^c \ell_i^{(k)}$ is more than $\theta$.
This ensures that the selected columns have accumulated ``energy'' at least $\theta$.
In this step, we have to carefully pick $\theta$, our {\it stopping threshold}.
This parameter essentially controls the {\it quality of the approximation}. 

In Section~\ref{sec:theta}, we provide some guidance on how the stopping parameter $\theta$ should be chosen.
Note that, if $\theta$ is such that $c < k$, we force $c = k$.
This is a necessary step that prevents the error in the approximation from ``blowing up"~(see Section~\ref{sec:theta}). 
The exact steps are given in Algorithm \ref{alg1}. 

\begin{algorithm}[!htb]
   \caption{\texttt{LeverageScoresSampler$(\A, k, \theta  )$}}
\begin{algorithmic}[1]
   \STATEx {\bf Input:} $\A \in \R^{m \times n}, ~k, ~\theta$
   \STATE Compute $\!\matV_k \!\in \!\R^{n \times k}\!$ (top $k$ right sing. vectors of $\matA$)
   \STATEx {\textbf{for} $i = 1, 2, \dots, n$}
   \STATE ~~~~$\ell_i^{(k)} =\left\|[\matV_k]_{i,:} \right\|_2^2$
   \STATEx {\textbf{end for}}
   \STATEx {without loss of generality, let $\ell_i^{(k)}$'s be sorted:
   $$ \ell_1^{(k)} \ge \dots \ge \ell_i^{(k)} \ge \ell_{i+1}^{(k)} \ge \dots \ge \ell_n^{(k)}. $$}
   \STATE Find index $ c \in \left\{1, \dots, n\right\}$ such that:
   $$c = \argmin_{c} \left( \sum_{i=1}^{c} \ell_i^{(k)} > \theta \right). $$
   \STATE If $c < k,$ set $c = k$.
   \STATEx {\bf Output:} $\matS \in \R^{n \times c}$ s.t. $\matA \matS$ has the top $c$ columns of~$\matA$.
\end{algorithmic}
\label{alg1}
 \end{algorithm}

Algorithm~\ref{alg1} requires $O(\min\{m,n\} m n)$ arithmetic operations. 
In Section~\ref{sec:extensions}
we discuss modifications to this algorithm which improve the running time. 

\section{Approximation guarantees}\label{sec:approx}
Our main technical innovation is a bound on the approximation error of  Algorithm 1 in regard to the CSSP; 
the proof of the following theorem can be found in Section~\ref{sec:proofs}.

\begin{theorem}\label{thm1}
Let 
$\theta = k- \varepsilon,$ for some $\varepsilon \in (0,1)$, and 
let $\matS$ be the $n \times c$ output  sampling matrix of Algorithm 1.
Then, for $\matC=\matA \matS$ and $\xi=\{2,\text{F}\}$, we have 
$ \XNormS{\matA - \matC \pinv{\matC}\matA} < \left(1-\varepsilon\right)^{-1} \cdot \XNormS{\matA - \matA_k}. $
\end{theorem}
Choosing $\varepsilon \in (0, 1/2)$ implies $(1-\varepsilon)^{-1} \le 1 + 2 \varepsilon$ and, 
hence, we have a relative-error approximation:
$$ \XNormS{\matA - \matC \pinv{\matC}\matA} < (1 + 2 \varepsilon) \cdot \XNormS{\matA - \matA_k}.$$

\subsection{Bounding the number of sampled columns}
\label{sec:bounding_c}
Algorithm 1 extracts at least $c\ge k$ columns of $\matA$.
However, an upper bound on the number of output columns $c$ is not immediate.
We study such upper bounds below. 

From Theorem~\ref{thm1}, it is clear that the stopping parameter $\theta = k-\varepsilon$ directly controls the number of output columns $c$.
This number, extracted for a specific error requirement $\varepsilon$, depends on the decay of the leverage scores.
For example, if the leverage scores decay fast, then we intuitively expect  $\sum_{i=1}^{c} \ell_i^{(k)}=k-\varepsilon$ to be achieved for a ``small'' $c$.

Let us for example consider a case where the leverage scores follow an extremely fast decay:
\begin{align*}
\ell_{1}^{(k)}&=k-2k\cdot \varepsilon,\\
\ell_{2}^{(k)}&=\ldots = \ell_{2k}^{(k)} = \varepsilon,\\
\ell_{2k+1}^{(k)}&=\ldots = \ell_{n}^{(k)} = \frac{\varepsilon}{n-2k}.
\end{align*}
Then, in this case
$\sum_{i=1}^{2k} \ell_i^{(k)}=k-\varepsilon,$
and Algorithm 1 outputs the $c=2k$ columns of $\matA$ that correspond to the $2k$ largest  leverage scores.
Due to Theorem \ref{thm1}, this subset of columns ${\bf C}\in\mathbb{R}^{n\times 2k}$ comes with the following guarantee: 
$$ \XNormS{\matA - \matC \pinv{\matC}\matA} < \frac{1}{1-\varepsilon} \cdot \XNormS{\matA - \matA_k}.$$
Hence, from the above example, we expect that, when the leverage scores decay fast, 
a small number of columns of $\matA$ will offer a good approximation of the form $\matC \matC^\dagger\matA$. 

However, in the worst case  Algorithm 1 can output a number of columns $c$ that
can be as large as $\Omega(n)$. 
To highlight this subtle point, consider
the case where the leverage scores are uniform
$\ell_i^{(k)} = \frac{k}{n}.$
Then, one can easily observe that if we want to achieve an error of $\varepsilon$ according to Theorem \ref{thm1}, we have to set $\theta = k-\varepsilon$.
This directly implies that we need to sample
 $c > (n/k) \theta$ columns.
Hence, if $\varepsilon=o(1),$ then, 
$$c\ge (n/k) \theta= (1-\varepsilon/k) n= \Omega(n).$$
Hence, for $\varepsilon \rightarrow 0$ we have $c \rightarrow n,$ 
which makes the result of Theorem \ref{thm1} trivial.

We argued above that when the leverage scores decay is ``fast'' then a good approximation
is to be expected with a ''small'' c. We make this intuition precise below~\footnote{We chose
to analyze in detail the case where the leverage scores follow a power law decay; other models
for the leverage scores, example, exponential decay, are also interesting, and will be the subject
of the full version of this work.}. 
The next theorem considers the case where
the leverage scores follow a power-law decay; the proof can be found in Section~\ref{sec:proofs}. 

\begin{theorem}\label{thm2}
Let the leverage scores follow a power-law decay  with exponent $\alpha_k = 1 + \eta$, for $\eta > 0$:
$$\ell_i^{(k)} = \frac{\ell_1^{(k)}}{i^{\alpha_k}}.$$
Then, if we set the stopping parameter to $\theta = k-\varepsilon,$ for some $\varepsilon$ with $0 < \varepsilon < 1,$
the number of sampled columns in ${\bf C}={\bf A}{\bf S}$ that Algorithm 1 outputs is
$$
c = \max\left\{ \left(\frac{2k}{\varepsilon}\right)^{\frac{1}{1+\eta}}-1, \;\;\;\left(\frac{2k}{\eta \cdot \varepsilon}\right)^{\frac{1}{\eta}} -1, \;\;\; k\right\},
$$ 
and  ${\bf C}$ achieves the following approximation error
$$ \XNormS{\matA - \matC \pinv{\matC}\matA} < \frac{1}{1-\varepsilon} \cdot \XNormS{\matA - \matA_k}, ~~\text{for $\xi=\{2,\mathrm{F}\}$} .$$
\end{theorem}

\subsection{Theoretical comparison to state of the art}
We compare the number of chosen columns $c$ in Algorithm 1 to the number of columns chosen in the randomized leverage scores sampling  
case~\cite{DMM08}.
The algorithm of  \cite{DMM08} requires 
$$c = O(k \log k / \varepsilon^2)$$ 
columns for a relative-error bound with respect to the Frobenius error in the CSSP:
$$ \FNormS{\matA - \matC \pinv{\matC}\matA} \le (1+\varepsilon) \|\matA-\matA_k\|_{\text{F}}^2.$$
Assuming the leverage scores follow a power-law decay, Algorithm 1 requires fewer columns for the same accuracy $\varepsilon$ when: 
$$ \max\left\{ \left(\frac{2k}{\varepsilon}\right)^{\frac{1}{1+\eta}}, ~\left(\frac{2k}{\eta \cdot \varepsilon}\right)^{\frac{1}{\eta}} \right\} <  C \cdot \frac{k\log k }{\varepsilon^2},$$
where $C$ is an absolute constant. 
Hence, under the power law decay,  Algorithm 1 offers provably a  matrix approximation similar or better than \cite{DMM08}.

Let us now compare the performance of Algorithm~\ref{alg1} with the results in~\cite{BDM11a}, which are the current
state of the art for the CSSP.  Theorem 1.5
in~\cite{BDM11a} provides a randomized algorithm which selects 
$$c = \frac{2k }{\varepsilon}(1+o(1))$$ columns in $\matC$
such that 
$$ \FNormS{\matA - \matC \pinv{\matC}\matA} < (1+\varepsilon) \cdot \FNormS{\matA - \matA_k}$$ 
holds in expectation.
This result is in fact optimal, up to a constant 2, since there is a lower bound indicating that such a relative error approximation 
is not possible unless 
$$c = k/\varepsilon,
$$ 
(see Section 9.2 in~\cite{BDM11a}). The approximation bound of Algorithm~\ref{alg1}  
is indeed better than the upper/lower bounds in~\cite{BDM11a} for any $\eta > 1$. We should note here that the lower bound
in~\cite{BDM11a} is for general matrices; however, the upper bound of Theorem~\ref{thm2} is applied to a specific class of matrices
whose leverage scores follow a power law decay. 

Next, we compare the spectral norm bound of Theorem~\ref{thm2} to the spectral norm bound of Theorem 1.1 in~\cite{BDM11a},
which indicates that there exists a deterministic algorithm selecting $c > k$ columns with error 
$$ \TNormS{\matA - \matC \pinv{\matC}\matA} < O\left(n/c \right) \cdot \TNormS{\matA - \matA_k}.$$
This upper bound is also tight, up to constants, since~\cite{BDM11a} provides a matching lower bound. Notice that
a relative error upper bound requires 
$$c = \Omega\left( n/(1+\varepsilon) \right)$$
in the general case. 
However, under the power law assumption
in Theorem~\ref{thm2}, we provide such a relative error bound with asymptotically fewer columns.
To our best knowledge, fixing $\eta$ to a constant, 
this is the \emph{first} relative-error bound for the spectral norm version of the CSSP
with 
$$c = \text{poly}(k,1/\varepsilon)$$
columns. 

\section{Experiments}
In this section, we first provide evidence that power law decays are prevalent in real-world data sets.
Then, we investigate the empirical performance of Algorithm~\ref{alg1} on real and synthetic data sets.

%


\begin{table*}[!htb]
\centering
\begin{scriptsize}
\begin{tabular}{c c c c c c c c c c c} \toprule
\multicolumn{1}{c}{Dataset} & \phantom{a} & \multicolumn{1}{c}{$m \times n$} & \phantom{a}  & \multicolumn{1}{c}{Description} & \phantom{a} & \multicolumn{1}{c}{Dataset} & \phantom{a} & \multicolumn{1}{c}{$m \times n$} & \phantom{a}  & \multicolumn{1}{c}{Description} \\
\cmidrule{1-1} \cmidrule{3-3} \cmidrule{5-5} \cmidrule{7-7}  \cmidrule{9-9} \cmidrule{11-11} 
Amazon & & $262111 \times 262111$ & & Purchase netw.~\cite{leskovec2009snap} & & Citeseer & & $723131 \times 723131$ & & Citation netw.~\cite{kunegis2013konect} \\
4square & & $106218 \times 106218$ & & Social netw.~\cite{Zafarani+Liu:2009} & & Github & & $56519 \times 120867$ & & Soft. netw. ~\cite{kunegis2013konect} \\
Gnutella & & $62586 \times 62586$ & & P2P netw.~\cite{leskovec2009snap} & & Google & & $875713 \times 875713$ & & Web conn. ~\cite{kunegis2013konect} \\
Gowalla & & $875713 \times 875713$ & & Social netw.~\cite{kunegis2013konect} & & LJournal & & $4847571 \times 4847571$ & & Social netw. ~\cite{leskovec2009snap} \\
Slashdot & & $82168 \times 82168$ & & Social netw.~\cite{leskovec2009snap} & & NIPS & & $12419 \times 1500$ & & Word/Docs ~\cite{Bache+Lichman:2013} \\
Skitter & & $1696415 \times 1696415$ & & System netw.~\cite{kunegis2013konect} & & CT slices & & $386 \times 53500$ & & CT images ~\cite{Bache+Lichman:2013} \\
Cora & & $23166 \times 23166$ & & Citation netw.~\cite{kunegis2013konect} & & Writer & & $81067 \times 42714$ & & Writers/Works ~\cite{kunegis2013konect} \\
Youtube & & $1134890 \times 1134890$ & & Video netw.~\cite{leskovec2009snap} & & YT groups & & $94238 \times 30087$ & & Users/Groups ~\cite{kunegis2013konect} \\
\bottomrule
\end{tabular}
\caption{Summary of datasets used in the experiments of Subsection~\ref{sec:exp_plaw}} \label{tbl:datasets2}

\end{scriptsize}
\end{table*}

Our experiments are not meant to be exhaustive; however, they provide clear evidence that:
$(i)$ the leverage scores of real world matrices indeed follow ``sharp'' power law decays; and
$(ii)$ deterministic leverage score sampling in such matrices is particularly effective.

\begin{figure*}[!htb]
\centering \includegraphics[width=0.9\textwidth]{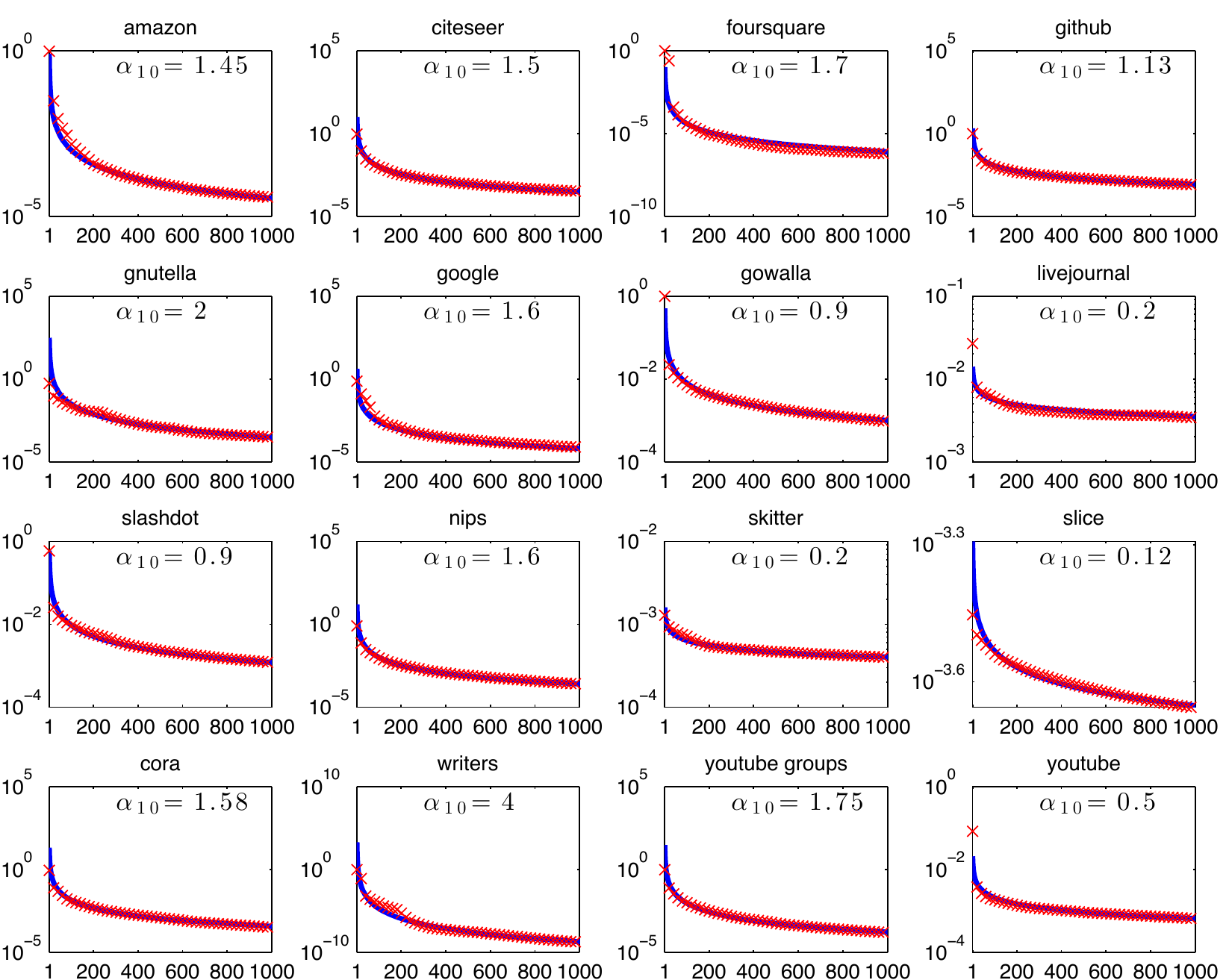}
\caption{We plot the top $1,000$ leverage scores for 16 different data sets, obtained through $\mathbf{V}_k$ for $k=10$.
The plots are in logarithmic scale.
For each data-set, we plot a fitting power-law curve $\beta \cdot x^{-\alpha_k}$.
The exponent is listed on each figure as $\alpha_{10}$.
The leverage scores are plotted with a red $\times$  marker, and the fitted curves are denoted with a solid blue line.
We observe that the power law fit offers a good approximation of the true leverage scores.
We further observe that many data sets exhibit sharp decays ($\alpha_k>1$), while only a few have leverage scores that decay slowly ($\alpha_k<1$).
}
\label{fig:plaw_fits}
\end{figure*}

\begin{figure*}[!htb]
\centerline{\includegraphics[width=0.9\textwidth]{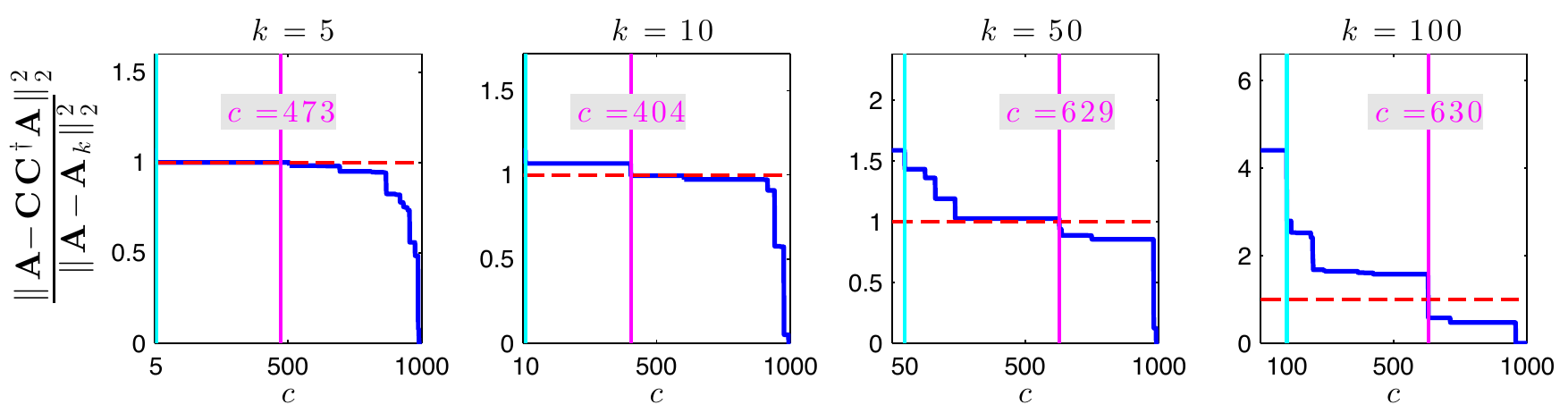}}
\caption{
Nearly-uniform leverage scores case:
Here, we plot as a blue curve the relative error 
 ratio $\|\matA - \matC \matC^\dagger\matA\|_2^2/ \|\matA - \matA_k\|_2^2$  achieved by Algorithm 1 as a function of the output columns $c$.
The leftmost vertical cyan line corresponds to the point where $k=c$. When $c<k$ the output error can be large; this justifies why we enforce the algorithm to output $c\ge k$ columns.
The rightmost vertical magenta line indicates the point where the $c$ sampled columns offer as good an approximation as that of the best rank-$k$ matrix ${\bf A}_k$.
 }\label{fig1}
\end{figure*}

\subsection{Power-law decays in real data sets}
\label{sec:exp_plaw}
We demonstrate the leverage score decay behavior of many real-world data sets.
These range from social networks and product co-purchasing matrices 
to document-term bag-of-words data sets, citation networks, and  medical imaging samples.
Their dimensions vary from thousands to millions of variables. The data-set description is given in Table~\ref{tbl:datasets2}.

In Figure~\ref{fig:plaw_fits}, we plot the top $1,000$ leverage scores extracted from the matrix of the right top-$k$ singular vectors ${\bf V}_k$.
In all cases we set $k=10$.\footnote{We performed various experiments for larger $k$, {\it e.g.,} $k=30$ or $k=100$ (not shown due to space limitations). 
We found that as we move towards higher $k$, we observe a ``smoothing" of the speed of decay.
This is to be expected, since for the case of $k = \rank(\matA)$ all leverage scores are equal.
}
For each dataset, we plot a fitting power-law curve of the form $\beta \cdot x^{-\alpha_k}$, where $\alpha_k$ is the exponent of interest.

We can see from the plots that a power law indeed seems to closely match the behavior of the top leverage scores.
What is more interesting is that for many of our data sets we observe a decay exponent of 
$\alpha_k>1$: this is the regime where deterministic sampling is expected to perform well.
It seems that these sharp decays are naturally present in many real-world data sets.

We would like to note that as we move to smaller scores (i.e., after the $10,000$-th score), 
we empirically observe that the leverage scores tail usually decays much faster than a power law.
This only helps the bound of Theorem 2.


\subsection{Synthetic Experiments}
\label{sec:exp_synthetic}
In this subsection, we are interested in understanding the performance of Algorithm~\ref{alg1} on matrices with {\it (i)} uniform and {\it (ii)} power-law decaying leverage scores.

To generate matrices with a prescribed leverage score decay, we use the implementation of
~\cite{ipsen2012effect}.
Let $\matV_k \in \R^{n \times k}$ denote the matrix we want to construct, for some $k< n$. 
Then, \cite{ipsen2012effect} provides algorithms to generate tall-and-skinny orthonormal matrices with specified row norms (i.e., leverage scores). 
Given the $\matV_k$ that is the output of the matrix generation algorithm in \cite{ipsen2012effect},
we run a basis completion algorithm to find the perpendicular matrix $\matV_k^{\perp} \in \R^{n \times (n-k)}$ such that
$\matV_k\transp \matV_k^{\perp} = {\bf 0}_{k \times (n-k)}.$
Then, we create an $n \times n$ matrix 
$\matV = [\matV_k \matV_k^{\perp}] $
where the first $k$ columns of $\matV$ are the columns of $\matV_k$ and the rest $n-k$ columns are the columns of $\matV_k^{\perp}$; hence, $\matV$ is a full orthonormal basis. 
Finally we generate $\matA \in \R^{m \times n}$ as $ \matA = \matU \matSig \matV\transp;$
where $\matU \in \mathbb{R}^{m \times m}$ is any orthonormal matrix, and $\matSig \in \mathbb{R}^{m \times n}$ any diagonal matrix
with $\min\{m,n\}$ positive entries along the main diagonal. 
Therefore, $\matA = \matU \matSig \matV\transp$ is the full SVD of $\matA$ with
leverage scores equal to the squared $\ell_2$-norm of the rows of $\matV_k$.
In our experiments, we pick $\matU$ as an orthonormal
basis for an $m \times m$ matrix where each entry is chosen i.i.d. from the Gaussian distribution. Also, 
$\matSig$ contains $\min\{m,n\}$ positive entries (sorted) along its main diagonal, where each entry was chosen i.i.d. from the Gaussian distribution.

\begin{figure*}[~!htb]
\centerline{\includegraphics[width=0.9\textwidth]{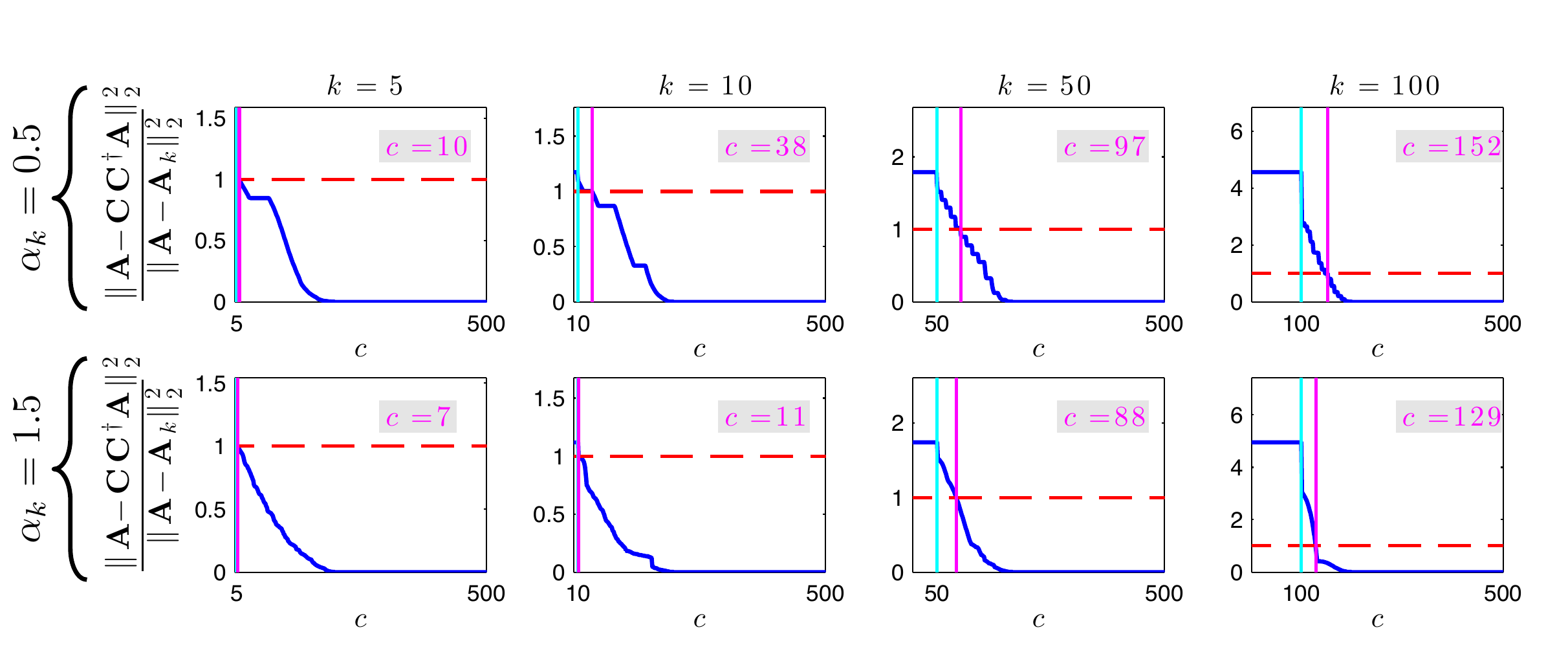}}
\caption{Power-law decaying leverage scores case:
We choose two power-law exponents: $\alpha_{k} = 0.5$ and $\alpha_{k} =1.5$.
In the first row we plot the relative error of Algorithm 1  vs. $c$ for the first decay profile, and the second row is the error performance of Algorithm 1 for the second, sharpest decay profile.
The vertical cyan line corresponds to the point where $k=c$,
and the vertical magenta line indicates the point where the $c$ sampled columns offer a better approximation compared to the best rank-$k$ matrix ${\bf A}_k$.
 }
 \label{fig2}
\end{figure*}

\subsubsection{ Nearly-uniform scores}
We set the number of rows to $m=200$ and the number of columns to $n=1000$ and construct
$\matA  = \matU \matSig \matV\transp \in \R^{m \times n}$ as described above.
The row norms of $\matV_k$ are chosen as follows: First, all row norms are chosen equal to $k/n$, for some fixed $k$.
Then, we introduce a small perturbation to avoid singularities: for every other pair of rows we add $\beta \in \mathcal{N}\left(0,1/100\right)$ to a row norm and subtract the same $\beta$ from the other row norm -- hence the sum of $\ell_i^{(k)}$ equals to $k$.

We set $k$ to take the values $\lbrace 5, 10, 50, 100 \rbrace $ and for each $k$ we choose:
$c = \{1,2,\ldots,1000\}$.
We present our findings in Figure~\ref{fig1},  where
we plot the relative error achieved 
$\frac{ \|\matA - \matC \matC^\dagger\matA\|_2^2}{\|\matA - \matA_k\|_2^2},$
where the $n\times c$ matrix ${\bf C}$ contains the first $c$ columns of $\matA$ that correspond to the $k$ largest leverage scores of ${\bf V}_k$, as sampled by Algorithm~\ref{alg1}.
Then, the leftmost vertical cyan line corresponds to the point where $k=c$,
and the rightmost vertical magenta line indicates the point where the $c$ sampled columns achieve an error of $\|\matA - \matA_k\|_2^2$, where $\matA_k$ is the best rank-$k$ approximation.

In the plots of Figure~\ref{fig1}, we see that as we move to larger values of $k$, if we wish to achieve an error of $ \|\matA - \matC \matC^\dagger\matA\|_2^2\approx \|\matA - \matA_k\|_2^2$, then we need to keep in ${\bf C}$, approximately almost half the columns of $\matA$.
This agrees with the uniform scores example that we showed earlier in Subsection \ref{sec:bounding_c}. 
However, we  observe that
Algorithm 1 can obtain a moderately small relative error, with significantly smaller $c$.
See for example the case where $k=100$; then, 
$c\approx 200$ sampled columns suffice for a relative error approximately equal to $2$, {\it i.e.,} $ \|\matA - \matC \matC^\dagger\matA\|_2^2\approx 2\cdot \|\matA - \matA_k\|_2^2$.
This indicates that our analysis could be loose in the general case.

\subsubsection{Power-law decay}
In this case, our synthetic eigenvector matrices ${\bf V}_k$ have leverage scores that follow a power law decay.
We choose two power-law exponents: $\alpha_{k} = 0.5$ and $\alpha_{k} = 1.5$.
Observe that the latter complies with Theorem~\ref{thm2}, that predicts the near optimality of leverage score sampling under such decay.

In the first row of Figure \ref{fig2}, we plot the relative error vs. the number of output columns $c$ of Algorithm 1 for $\alpha_k=0.5$.
Then, in the second row of Figure \ref{fig2}, we plot the relative error vs. the number of output columns $c$ of Algorithm 1 for $\alpha_k=1.5$.
The blue line represents the relative error in terms of spectral norm.
We can see that the performance of Algorithm 1 in the case of the fast decay is surprising: $c \approx 1.5\cdot k$ suffices for an approximation as good as of that of the best rank-$k$ approximation. This confirms the approximation performance in Theorem~\ref{thm2}.

\newcommand{\ra}[1]{\renewcommand{\arraystretch}{#1}}

\subsection{Comparison with other techniques}
\label{sec:exp_real}
We will now compare the proposed algorithm to state of the art approaches for the CSSP,  both for $\xi = 2$ and $\xi = \text{F}$.
We report results for the errors $\XNormS{\matA - \matC \matC^\dagger\matA} / \XNormS{\matA - \matA_k}$.
A comparison of the running time complexity of those algorithms is out of the scope of our experiments. 

Table \ref{tbl:datasets} contains a brief description of the datasets used in our experiments. We employ the datasets used in~\cite{GM13}, which presents exhaustive experiments for matrix approximations obtained through randomized leverage scores sampling. 

\begin{table*}[!htb]
\centering
\ra{1.3}
\begin{small}
\begin{tabular}{l c c c c c c c c} \toprule
\multicolumn{1}{c}{Dataset} & \phantom{ab} & \multicolumn{1}{c}{$m$} & \phantom{ab} & \multicolumn{1}{c}{$n$} &  \phantom{ab}  & \multicolumn{1}{c}{$\text{rank}(\mathbf{A})$} &  \phantom{ab}  & \multicolumn{1}{c}{Description} \\
\cmidrule{1-1} \cmidrule{3-3} \cmidrule{5-5} \cmidrule{7-7}  \cmidrule{9-9}  
\multicolumn{1}{c}{\texttt{Protein}} & & \multicolumn{1}{c}{$357$} & & \multicolumn{1}{c}{$6621$} & & \multicolumn{1}{c}{$356$}  & & \multicolumn{1}{c}{Saccharomyces cerevisiae dataset}\\ 
\multicolumn{1}{c}{\texttt{SNPS}} & & \multicolumn{1}{c}{$46$} & & \multicolumn{1}{c}{$5523$} & & \multicolumn{1}{c}{$46$}  & & \multicolumn{1}{c}{Single Nucleotide - polymorphism dataset}\\ 
\multicolumn{1}{c}{\texttt{Enron}} & & \multicolumn{1}{c}{$3000$} & & \multicolumn{1}{c}{$3000$} & & \multicolumn{1}{c}{$2569$}  & & \multicolumn{1}{c}{A subgraph of the Enron email graph}\\ 
\bottomrule
\end{tabular}
\end{small}
\caption{Summary of datasets used in the experiments of Subsection \ref{sec:exp_real} \cite{GM13}} \label{tbl:datasets}
\end{table*}

\begin{table*}[!ht]
\centering
\ra{1.3}
\begin{scriptsize}
\begin{tabular}{l c c c c c c c c c c c} \toprule
\multicolumn{3}{c}{Model} & \phantom{ab} & \multicolumn{3}{c}{$\frac{\|\matA - \matC\matC^\dagger\matA\|_2}{\|\matA - \matA_k\|_2}$} & \phantom{ab} & \multicolumn{4}{c}{$\frac{\|\matA - \matC\matC^\dagger\matA\|_\text{F}}{\|\matA - \matA_k\|_\text{F}}$} \\
\cmidrule{1-3} \cmidrule{5-7} \cmidrule{9-12}
\multicolumn{1}{c}{} & \multicolumn{1}{c}{$k$} & \multicolumn{1}{c}{$c$} & & \multicolumn{1}{c}{\cite{DMM08}} & \multicolumn{1}{c}{\cite{Gol65}} & \multicolumn{1}{c}{This work} & & \multicolumn{1}{c}{\cite{BDM11a}} & \multicolumn{1}{c}{\cite{DMM08}} & \multicolumn{1}{c}{\cite{Gol65}} & \multicolumn{1}{c}{This work} \\
\midrule
\multirow{10}{*}{\texttt{Protein}} & \multirow{5}{*}{$50$} & \multicolumn{1}{c}{$51$} & & \multicolumn{1}{c}{$\textcolor[rgb]{0.4,0.1,0}{\mathbf{1.7334}}$} & \multicolumn{1}{c}{$2.2621$} & \multicolumn{1}{c}{$2.6809$} & & \multicolumn{1}{c}{$1.1068$} & \multicolumn{1}{c}{$1.0888$} & \multicolumn{1}{c}{$1.1017$} & \multicolumn{1}{c}{$\textcolor[rgb]{0.4,0.1,0}{\mathbf{1.1000}}$} \\ 
& & \multicolumn{1}{c}{$118$} & & \multicolumn{1}{c}{$\textcolor[rgb]{0.4,0.1,0}{\mathbf{1.3228}}$} & \multicolumn{1}{c}{$1.4274$} & \multicolumn{1}{c}{$2.2536$} & & \multicolumn{1}{c}{$0.9344$} & \multicolumn{1}{c}{$\textcolor[rgb]{0.4,0.1,0}{\mathbf{0.9233}}$} & \multicolumn{1}{c}{$0.9258$} & \multicolumn{1}{c}{$0.9259$} \\ 
& & \multicolumn{1}{c}{$186$} & & \multicolumn{1}{c}{$1.0846$} & \multicolumn{1}{c}{$\textcolor[rgb]{0.4,0.1,0}{\mathbf{1.0755}}$} & \multicolumn{1}{c}{$1.7357$} & & \multicolumn{1}{c}{$0.7939$} & \multicolumn{1}{c}{$0.7377$} & \multicolumn{1}{c}{$0.7423$} & \multicolumn{1}{c}{$\textcolor[rgb]{0.4,0.1,0}{\mathbf{0.7346}}$} \\ 
& & \multicolumn{1}{c}{$253$} & & \multicolumn{1}{c}{$\textcolor[rgb]{0.4,0.1,0}{\mathbf{0.9274}}$} & \multicolumn{1}{c}{$0.9281$} & \multicolumn{1}{c}{$1.3858$} & & \multicolumn{1}{c}{$0.6938$} & \multicolumn{1}{c}{$0.5461$} & \multicolumn{1}{c}{$0.5326$} & \multicolumn{1}{c}{$\textcolor[rgb]{0.4,0.1,0}{\mathbf{0.5264}}$} \\ 
& & \multicolumn{1}{c}{$320$} & & \multicolumn{1}{c}{$0.7899$} & \multicolumn{1}{c}{$\textcolor[rgb]{0.4,0.1,0}{\mathbf{0.7528}}$} & \multicolumn{1}{c}{$0.8176$} & & \multicolumn{1}{c}{$0.5943$} & \multicolumn{1}{c}{$0.2831$} & \multicolumn{1}{c}{$0.2303$} & \multicolumn{1}{c}{$\textcolor[rgb]{0.4,0.1,0}{\mathbf{0.2231}}$} \\ 
\cmidrule{2-12}
& \multirow{5}{*}{$100$} & \multicolumn{1}{c}{$101$} & & \multicolumn{1}{c}{$1.8568$} & \multicolumn{1}{c}{$\textcolor[rgb]{0.4,0.1,0}{\mathbf{1.8220}}$} & \multicolumn{1}{c}{$2.5666$} & & \multicolumn{1}{c}{$1.1789$} & \multicolumn{1}{c}{$\textcolor[rgb]{0.4,0.1,0}{\mathbf{1.1506}}$} & \multicolumn{1}{c}{$1.1558$} & \multicolumn{1}{c}{$1.1606$} \\ 
& & \multicolumn{1}{c}{$156$} & & \multicolumn{1}{c}{$\textcolor[rgb]{0.4,0.1,0}{\mathbf{1.3741}}$} & \multicolumn{1}{c}{$1.3987$} & \multicolumn{1}{c}{$2.4227$} & & \multicolumn{1}{c}{$0.9928$} & \multicolumn{1}{c}{$\textcolor[rgb]{0.4,0.1,0}{\mathbf{0.9783}}$} & \multicolumn{1}{c}{$0.9835$} & \multicolumn{1}{c}{$0.9820$} \\ 
& & \multicolumn{1}{c}{$211$} & & \multicolumn{1}{c}{$1.3041$} & \multicolumn{1}{c}{$\textcolor[rgb]{0.4,0.1,0}{\mathbf{1.1926}}$} & \multicolumn{1}{c}{$2.3122$} & & \multicolumn{1}{c}{$0.8182$} & \multicolumn{1}{c}{$0.8100$} & \multicolumn{1}{c}{$0.7958$} & \multicolumn{1}{c}{$\textcolor[rgb]{0.4,0.1,0}{\mathbf{0.7886}}$} \\ 
& & \multicolumn{1}{c}{$265$} & & \multicolumn{1}{c}{$\textcolor[rgb]{0.4,0.1,0}{\mathbf{1.0270}}$} & \multicolumn{1}{c}{$1.0459$} & \multicolumn{1}{c}{$2.0509$} & & \multicolumn{1}{c}{$0.6241$} & \multicolumn{1}{c}{$0.6004$} & \multicolumn{1}{c}{$0.5820$} & \multicolumn{1}{c}{$\textcolor[rgb]{0.4,0.1,0}{\mathbf{0.5768}}$} \\ 
& & \multicolumn{1}{c}{$320$} & & \multicolumn{1}{c}{$0.9174$} & \multicolumn{1}{c}{$\textcolor[rgb]{0.4,0.1,0}{\mathbf{0.8704}}$} & \multicolumn{1}{c}{$1.8562$} & & \multicolumn{1}{c}{$0.3752$} & \multicolumn{1}{c}{$0.3648$} & \multicolumn{1}{c}{$\textcolor[rgb]{0.4,0.1,0}{\mathbf{0.2742}}$} & \multicolumn{1}{c}{$0.2874$} \\ 
\midrule
\multirow{10}{*}{\texttt{SNPS}} & \multirow{5}{*}{$5$} & \multicolumn{1}{c}{$6$} & & \multicolumn{1}{c}{$\textcolor[rgb]{0.4,0.1,0}{\mathbf{1.4765}}$} & \multicolumn{1}{c}{$1.5030$} & \multicolumn{1}{c}{$1.5613$} & & \multicolumn{1}{c}{$1.1831$} & \multicolumn{1}{c}{$\textcolor[rgb]{0.4,0.1,0}{\mathbf{1.0915}}$} & \multicolumn{1}{c}{$1.1030$} & \multicolumn{1}{c}{$1.1056$} \\ 
& & \multicolumn{1}{c}{$12$} & & \multicolumn{1}{c}{$1.2601$} & \multicolumn{1}{c}{$\textcolor[rgb]{0.4,0.1,0}{\mathbf{1.2402}}$} & \multicolumn{1}{c}{$1.2799$} & & \multicolumn{1}{c}{$1.0524$} & \multicolumn{1}{c}{$0.9649$} & \multicolumn{1}{c}{$0.9519$} & \multicolumn{1}{c}{$\textcolor[rgb]{0.4,0.1,0}{\mathbf{0.9469}}$} \\ 
& & \multicolumn{1}{c}{$18$} & & \multicolumn{1}{c}{$1.0537$} & \multicolumn{1}{c}{$\textcolor[rgb]{0.4,0.1,0}{\mathbf{1.0236}}$} & \multicolumn{1}{c}{$1.1252$} & & \multicolumn{1}{c}{$1.0183$} & \multicolumn{1}{c}{$0.8283$} & \multicolumn{1}{c}{$\textcolor[rgb]{0.4,0.1,0}{\mathbf{0.8187}}$} & \multicolumn{1}{c}{$0.8281$} \\ 
& & \multicolumn{1}{c}{$24$} & & \multicolumn{1}{c}{$\textcolor[rgb]{0.4,0.1,0}{\mathbf{0.8679}}$} & \multicolumn{1}{c}{$0.9063$} & \multicolumn{1}{c}{$0.9302$} & & \multicolumn{1}{c}{$0.9537$} & \multicolumn{1}{c}{$0.6943$} & \multicolumn{1}{c}{$\textcolor[rgb]{0.4,0.1,0}{\mathbf{0.6898}}$} & \multicolumn{1}{c}{$0.6975$} \\ 
& & \multicolumn{1}{c}{$30$} & & \multicolumn{1}{c}{$\textcolor[rgb]{0.4,0.1,0}{\mathbf{0.7441}}$} & \multicolumn{1}{c}{$0.7549$} & \multicolumn{1}{c}{$0.8742$} & & \multicolumn{1}{c}{$0.9558$} & \multicolumn{1}{c}{$0.5827$} & \multicolumn{1}{c}{$\textcolor[rgb]{0.4,0.1,0}{\mathbf{0.5413}}$} & \multicolumn{1}{c}{$0.5789$} \\ 
\cmidrule{2-12}
& \multirow{5}{*}{$10$} & \multicolumn{1}{c}{$11$} & & \multicolumn{1}{c}{$1.6459$} & \multicolumn{1}{c}{$\textcolor[rgb]{0.4,0.1,0}{\mathbf{1.5206}}$} & \multicolumn{1}{c}{$1.6329$} & & \multicolumn{1}{c}{$1.2324$} & \multicolumn{1}{c}{$1.1708$} & \multicolumn{1}{c}{$1.1500$} & \multicolumn{1}{c}{$\textcolor[rgb]{0.4,0.1,0}{\mathbf{1.1413}}$} \\ 
& & \multicolumn{1}{c}{$16$} & & \multicolumn{1}{c}{$\textcolor[rgb]{0.4,0.1,0}{\mathbf{1.3020}}$} & \multicolumn{1}{c}{$1.4265$} & \multicolumn{1}{c}{$1.5939$} & & \multicolumn{1}{c}{$1.1272$} & \multicolumn{1}{c}{$1.0386$} & \multicolumn{1}{c}{$\textcolor[rgb]{0.4,0.1,0}{\mathbf{1.0199}}$} & \multicolumn{1}{c}{$1.0420$} \\ 
& & \multicolumn{1}{c}{$21$} & & \multicolumn{1}{c}{$1.2789$} & \multicolumn{1}{c}{$\textcolor[rgb]{0.4,0.1,0}{\mathbf{1.1511}}$} & \multicolumn{1}{c}{$1.1676$} & & \multicolumn{1}{c}{$1.0225$} & \multicolumn{1}{c}{$0.9170$} & \multicolumn{1}{c}{$\textcolor[rgb]{0.4,0.1,0}{\mathbf{0.8842}}$} & \multicolumn{1}{c}{$0.9011$} \\ 
& & \multicolumn{1}{c}{$25$} & & \multicolumn{1}{c}{$1.1022$} & \multicolumn{1}{c}{$\textcolor[rgb]{0.4,0.1,0}{\mathbf{1.0729}}$} & \multicolumn{1}{c}{$1.0935$} & & \multicolumn{1}{c}{$0.9838$} & \multicolumn{1}{c}{$0.8091$} & \multicolumn{1}{c}{$\textcolor[rgb]{0.4,0.1,0}{\mathbf{0.7876}}$} & \multicolumn{1}{c}{$0.8057$} \\ 
& & \multicolumn{1}{c}{$30$} & & \multicolumn{1}{c}{$0.9968$} & \multicolumn{1}{c}{$\textcolor[rgb]{0.4,0.1,0}{\mathbf{0.9256}}$} & \multicolumn{1}{c}{$1.0020$} & & \multicolumn{1}{c}{$0.8088$} & \multicolumn{1}{c}{$0.6636$} & \multicolumn{1}{c}{$\textcolor[rgb]{0.4,0.1,0}{\mathbf{0.6375}}$} & \multicolumn{1}{c}{$0.6707$} \\ 
\midrule
\multirow{21}{*}{\texttt{Enron}} & \multirow{5}{*}{$10$} & \multicolumn{1}{c}{$11$} & & \multicolumn{1}{c}{$2.2337$} & \multicolumn{1}{c}{$1.8320$} & \multicolumn{1}{c}{$\textcolor[rgb]{0.4,0.1,0}{\mathbf{1.7217}}$} & & \multicolumn{1}{c}{$1.1096$} & \multicolumn{1}{c}{$1.0992$} & \multicolumn{1}{c}{$1.0768$} & \multicolumn{1}{c}{$\textcolor[rgb]{0.4,0.1,0}{\mathbf{1.0704}}$} \\ 
& & \multicolumn{1}{c}{$83$} & & \multicolumn{1}{c}{$\textcolor[rgb]{0.4,0.1,0}{\mathbf{1.0717}}$} & \multicolumn{1}{c}{$1.0821$} & \multicolumn{1}{c}{$1.1464$} & & \multicolumn{1}{c}{$1.0123$} & \multicolumn{1}{c}{$0.9381$} & \multicolumn{1}{c}{$\textcolor[rgb]{0.4,0.1,0}{\mathbf{0.9094}}$} & \multicolumn{1}{c}{$0.9196$} \\ 
& & \multicolumn{1}{c}{$156$} & & \multicolumn{1}{c}{$0.8419$} & \multicolumn{1}{c}{$\textcolor[rgb]{0.4,0.1,0}{\mathbf{0.8172}}$} & \multicolumn{1}{c}{$0.8412$} & & \multicolumn{1}{c}{$1.0044$} & \multicolumn{1}{c}{$0.8692$} & \multicolumn{1}{c}{$\textcolor[rgb]{0.4,0.1,0}{\mathbf{0.8091}}$} & \multicolumn{1}{c}{$0.8247$} \\ 
& & \multicolumn{1}{c}{$228$} & & \multicolumn{1}{c}{$\textcolor[rgb]{0.4,0.1,0}{\mathbf{0.6739}}$} & \multicolumn{1}{c}{$0.6882$} & \multicolumn{1}{c}{$0.6993$} & & \multicolumn{1}{c}{$0.9984$} & \multicolumn{1}{c}{$0.8096$} & \multicolumn{1}{c}{$\textcolor[rgb]{0.4,0.1,0}{\mathbf{0.7311}}$} & \multicolumn{1}{c}{$0.7519$} \\ 
& & \multicolumn{1}{c}{$300$} & & \multicolumn{1}{c}{$0.6061$} & \multicolumn{1}{c}{$\textcolor[rgb]{0.4,0.1,0}{\mathbf{0.6041}}$} & \multicolumn{1}{c}{$0.6057$} & & \multicolumn{1}{c}{$1.0000$} & \multicolumn{1}{c}{$0.7628$} & \multicolumn{1}{c}{$\textcolor[rgb]{0.4,0.1,0}{\mathbf{0.6640}}$} & \multicolumn{1}{c}{$0.6837$} \\ 
\cmidrule{2-12}
& \multirow{5}{*}{$20$} & \multicolumn{1}{c}{$21$} & & \multicolumn{1}{c}{$2.1726$} & \multicolumn{1}{c}{$\textcolor[rgb]{0.4,0.1,0}{\mathbf{1.9741}}$} & \multicolumn{1}{c}{$2.1669$} & & \multicolumn{1}{c}{$1.1344$} & \multicolumn{1}{c}{$1.1094$} & \multicolumn{1}{c}{$\textcolor[rgb]{0.4,0.1,0}{\mathbf{1.0889}}$} & \multicolumn{1}{c}{$1.0931$} \\ 
& & \multicolumn{1}{c}{$91$} & & \multicolumn{1}{c}{$1.3502$} & \multicolumn{1}{c}{$\textcolor[rgb]{0.4,0.1,0}{\mathbf{1.3305}}$} & \multicolumn{1}{c}{$1.3344$} & & \multicolumn{1}{c}{$1.0194$} & \multicolumn{1}{c}{$0.9814$} & \multicolumn{1}{c}{$\textcolor[rgb]{0.4,0.1,0}{\mathbf{0.9414}}$} & \multicolumn{1}{c}{$0.9421$} \\ 
& & \multicolumn{1}{c}{$161$} & & \multicolumn{1}{c}{$1.0242$} & \multicolumn{1}{c}{$1.0504$} & \multicolumn{1}{c}{$\textcolor[rgb]{0.4,0.1,0}{\mathbf{1.0239}}$} & & \multicolumn{1}{c}{$0.9999$} & \multicolumn{1}{c}{$0.9004$} & \multicolumn{1}{c}{$\textcolor[rgb]{0.4,0.1,0}{\mathbf{0.8434}}$} & \multicolumn{1}{c}{$0.8484$} \\ 
& & \multicolumn{1}{c}{$230$} & & \multicolumn{1}{c}{$0.9099$} & \multicolumn{1}{c}{$0.9025$} & \multicolumn{1}{c}{$\textcolor[rgb]{0.4,0.1,0}{\mathbf{0.9006}}$} & & \multicolumn{1}{c}{$0.9730$} & \multicolumn{1}{c}{$0.8505$} & \multicolumn{1}{c}{$\textcolor[rgb]{0.4,0.1,0}{\mathbf{0.7655}}$} & \multicolumn{1}{c}{$0.7740$} \\ 
& & \multicolumn{1}{c}{$300$} & & \multicolumn{1}{c}{$0.8211$} & \multicolumn{1}{c}{$0.7941$} & \multicolumn{1}{c}{$\textcolor[rgb]{0.4,0.1,0}{\mathbf{0.7936}}$} & & \multicolumn{1}{c}{$0.9671$} & \multicolumn{1}{c}{$0.8037$} & \multicolumn{1}{c}{$\textcolor[rgb]{0.4,0.1,0}{\mathbf{0.6971}}$} & \multicolumn{1}{c}{$0.7087$} \\ 
\cmidrule{2-12}
& \multirow{5}{*}{$50$} & \multicolumn{1}{c}{$51$} & & \multicolumn{1}{c}{$2.6520$} & \multicolumn{1}{c}{$2.2788$} & \multicolumn{1}{c}{$\textcolor[rgb]{0.4,0.1,0}{\mathbf{2.2520}}$} & & \multicolumn{1}{c}{$1.1547$} & \multicolumn{1}{c}{$1.1436$} & \multicolumn{1}{c}{$\textcolor[rgb]{0.4,0.1,0}{\mathbf{1.1053}}$} & \multicolumn{1}{c}{$1.1076$} \\ 
& & \multicolumn{1}{c}{$113$} & & \multicolumn{1}{c}{$1.7454$} & \multicolumn{1}{c}{$\textcolor[rgb]{0.4,0.1,0}{\mathbf{1.6850}}$} & \multicolumn{1}{c}{$1.8122$} & & \multicolumn{1}{c}{$1.0350$} & \multicolumn{1}{c}{$1.0425$} & \multicolumn{1}{c}{$\textcolor[rgb]{0.4,0.1,0}{\mathbf{0.9902}}$} & \multicolumn{1}{c}{$0.9929$} \\ 
& & \multicolumn{1}{c}{$176$} & & \multicolumn{1}{c}{$\textcolor[rgb]{0.4,0.1,0}{\mathbf{1.3524}}$} & \multicolumn{1}{c}{$1.4199$} & \multicolumn{1}{c}{$1.4673$} & & \multicolumn{1}{c}{$0.9835$} & \multicolumn{1}{c}{$0.9718$} & \multicolumn{1}{c}{$\textcolor[rgb]{0.4,0.1,0}{\mathbf{0.8999}}$} & \multicolumn{1}{c}{$0.9011$} \\ 
& & \multicolumn{1}{c}{$238$} & & \multicolumn{1}{c}{$1.2588$} & \multicolumn{1}{c}{$\textcolor[rgb]{0.4,0.1,0}{\mathbf{1.2303}}$} & \multicolumn{1}{c}{$1.2450$} & & \multicolumn{1}{c}{$0.9607$} & \multicolumn{1}{c}{$0.9187$} & \multicolumn{1}{c}{$\textcolor[rgb]{0.4,0.1,0}{\mathbf{0.8251}}$} & \multicolumn{1}{c}{$0.8282$} \\ 
& & \multicolumn{1}{c}{$300$} & & \multicolumn{1}{c}{$1.2209$} & \multicolumn{1}{c}{$\textcolor[rgb]{0.4,0.1,0}{\mathbf{1.1014}}$} & \multicolumn{1}{c}{$1.1239$} & & \multicolumn{1}{c}{$0.9384$} & \multicolumn{1}{c}{$0.8806$} & \multicolumn{1}{c}{$\textcolor[rgb]{0.4,0.1,0}{\mathbf{0.7593}}$} & \multicolumn{1}{c}{$0.7651$} \\ 
\cmidrule{2-12}
& \multirow{5}{*}{$100$} & \multicolumn{1}{c}{$101$} & & \multicolumn{1}{c}{$2.2502$} & \multicolumn{1}{c}{$\textcolor[rgb]{0.4,0.1,0}{\mathbf{2.2145}}$} & \multicolumn{1}{c}{$2.2721$} & & \multicolumn{1}{c}{$1.1938$} & \multicolumn{1}{c}{$1.1805$} & \multicolumn{1}{c}{$\textcolor[rgb]{0.4,0.1,0}{\mathbf{1.1223}}$} & \multicolumn{1}{c}{$1.1238$} \\ 
& & \multicolumn{1}{c}{$151$} & & \multicolumn{1}{c}{$2.2399$} & \multicolumn{1}{c}{$\textcolor[rgb]{0.4,0.1,0}{\mathbf{1.8677}}$} & \multicolumn{1}{c}{$1.8979$} & & \multicolumn{1}{c}{$1.0891$} & \multicolumn{1}{c}{$1.1122$} & \multicolumn{1}{c}{$\textcolor[rgb]{0.4,0.1,0}{\mathbf{1.0357}}$} & \multicolumn{1}{c}{$1.0393$} \\ 
& & \multicolumn{1}{c}{$201$} & & \multicolumn{1}{c}{$1.7945$} & \multicolumn{1}{c}{$1.6350$} & \multicolumn{1}{c}{$\textcolor[rgb]{0.4,0.1,0}{\mathbf{1.6332}}$} & & \multicolumn{1}{c}{$1.0236$} & \multicolumn{1}{c}{$1.0631$} & \multicolumn{1}{c}{$\textcolor[rgb]{0.4,0.1,0}{\mathbf{0.9646}}$} & \multicolumn{1}{c}{$0.9664$} \\ 
& & \multicolumn{1}{c}{$250$} & & \multicolumn{1}{c}{$1.6721$} & \multicolumn{1}{c}{$\textcolor[rgb]{0.4,0.1,0}{\mathbf{1.5001}}$} & \multicolumn{1}{c}{$1.5017$} & & \multicolumn{1}{c}{$0.9885$} & \multicolumn{1}{c}{$1.0026$} & \multicolumn{1}{c}{$\textcolor[rgb]{0.4,0.1,0}{\mathbf{0.9025}}$} & \multicolumn{1}{c}{$0.9037$} \\ 
& & \multicolumn{1}{c}{$300$} & & \multicolumn{1}{c}{$1.3946$} & \multicolumn{1}{c}{$\textcolor[rgb]{0.4,0.1,0}{\mathbf{1.3711}}$} & \multicolumn{1}{c}{$1.3847$} & & \multicolumn{1}{c}{$0.9485$} & \multicolumn{1}{c}{$0.9672$} & \multicolumn{1}{c}{$\textcolor[rgb]{0.4,0.1,0}{\mathbf{0.8444}}$} & \multicolumn{1}{c}{$0.8467$} \\ 
\bottomrule
\end{tabular}
\end{scriptsize}
\label{tbl:results1}
\caption{We present the performance of Algorithm 1 as compared to the state of the art in CSSP.
We run experiments on 3 data sets described in the above table, for various values of $k$ and $c$.
The performance of Algorithm 1, especially in terms of the Frobenius norm error, is very close to optimal, while at the same time similar, if not better, to the performance of the more sophisticated algorithms of the comparison.
}
\end{table*}

\subsubsection{List of comparison algorithms}
 We compare Algorithm~\ref{alg1} against three methods for the CSSP. First, the authors in~\cite{BDM11a} present a near-optimal deterministic algorithm, as described  in Theorem 
1.2 in ~\cite{BDM11a}.
Given $\matA, ~k < \rank(\matA)$ and $c > k$, the proposed algorithm selects 
$\tilde{c} \le c$ columns of $\matA$ in $\matC \in \R^{m \times \tilde{c}}$ with
$$
\FNorm{ \matA - \matC \pinv{\matC} \matA } \le \left(1 + \left(1-\sqrt{k/c} \right)^{-1} \right)\FNorm{\matA - \matA_k}.
$$

Second, in ~\cite{Gol65}, the authors present a deterministic pivoted QR algorithm such that: 
$$ \TNorm{ \matA - \matC \pinv{\matC} \matA } \le \left( 1 +\sqrt{n-k} \cdot 2^k  \right) \TNorm{\matA - \matA_k}. $$
This bound was proved in~\cite{GE96}. 
In our tests, we use the \texttt{qr}($\cdot$) built-in Matlab function, where one can select $c=k$ columns of $\matA$ as: 
\begin{align}
[\mathbf{Q},~\mathbf{R}, \boldsymbol{\pi}] &= \texttt{qr}(\matA,0); \nonumber \quad \matC = \matA_{:,\boldsymbol{\pi}_{1:c}},
\end{align} 
where $\matA = \mathbf{Q} \mathbf{R}$, $\mathbf{Q} \in \R^{m \times n}$ contains orthonormal columns,
$\matR \in \R^{n \times n}$ is upper triangular, 
and $\boldsymbol{\pi}$ is a permutation information vector such that $\matA_{:, \boldsymbol{\pi}} = \mathbf{Q}\mathbf{R}$. 

Third, we also consider the randomized leverage-scores sampling method with replacement, presented in~\cite{DMM08}. According to this work and given $\matA, k < \rank(\matA),$ and $c = \Omega(k \log k)$, the bound provided by the algorithm is 
$$ \FNorm{ \matA - \matC \pinv{\matC} \matA } \le \left( 1 + O\left( \sqrt{k \log k / c}  \right) \right) \FNorm{\matA - \matA_k}, $$
which holds only with constant probability. 
In our experiments, we use the software tool developed in~\cite{ipsen2012effect} for the randomized sampling step.

We use our own Matlab implementation for each of these approaches.  For \cite{DMM08}, we execute $10$ repetitions and report the one that minimizes the approximation error. 

\subsubsection{Performance Results}
Table 3 contains a subset of our results; a complete set of results is reserved for an extended version of this work. 
We observe that the performance of Algorithm~\ref{alg1} is particularly appealing; particularly, it is almost 
 as good as randomized leverage scores sampling in almost all cases - when randomized
 sampling is better the difference is often on the first or second decimal digit. 

Figure~\ref{fig:plaw_fits_2} shows the leverage scores for the three matrices used in our experiments. 
We see that although the decay for the first data sets does not fit a ``sharp'' power law as it is required in Theorem \ref{thm2}, the 
performance of the algorithm is still competitive in practice. 
Interestingly, we do observe good performance compared to the other algorithms for the third data set (Enron).
For this case, the power law decay fits the decay profile needed to establish the near optimality of Algorithm 1.

\begin{figure}[h]
\centering \includegraphics[width=1\columnwidth]{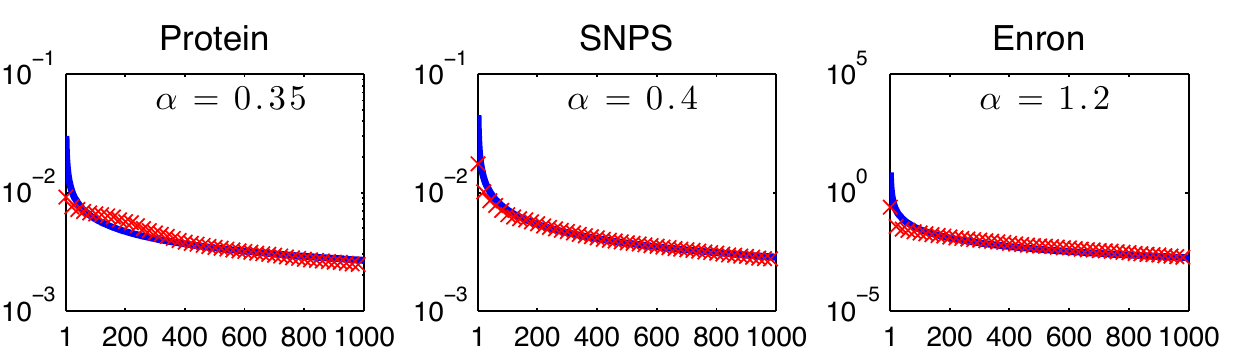}
\vspace{0.1cm}
\caption{
The plots are for $k=10$ and are in logarithmic scale. 
The exponent is listed on each figure as $\alpha$.
The leverage scores are plotted with a red $\times$  marker, and the fitted curves are denoted with a solid blue line.
}
\label{fig:plaw_fits_2}
\end{figure}

\section{Related work}\label{sec:related}
We give a quick overview of several column subset selection algorithms, both deterministic and randomized.

One of the first deterministic results regarding the CSSP
goes back to the seminal work of Gene Golub on pivoted QR factorizations~\cite{Gol65}.
Similar algorithms have been developed in~\cite{Gol65, HP92, CH94, CI94, GE96, Tyr96, Ste99, BQ98a, Pan00};
see also~\cite{BMD09a} for a recent survey. The best of these algorithms is the so-called \emph{Strong Rank-revealing QR} (Strong RRQR)
algorithm in~\cite{GE96}: Given
$\matA$, $c = k,$ and constant $f \ge 1,$ Strong RRQR requires $O(m n k \log_{f} n)$ arithmetic operations to find
$k$ columns of $\matA$ in $\matC \in \R^{m \times k}$ that satisfy
$$ \TNorm{\matA - \matC \pinv{\matC}\matA} \le \left(1 + f^2 \sqrt{k(n-k)+1} \right) \cdot \TNorm{\matA - \matA_k}. $$

As discussed in Section~\ref{sec:intro},
\cite{Jol72} suggests column sampling with the largest corresponding leverage scores.
A related result in~\cite{Tyr96} suggests column sampling through selection over $\matV_k\transp$ with Strong RRQR.
Notice that the leverage scores sampling approach is similar, but the column selection is based on the largest Euclidean norms of the  columns of $\matV_k\transp$.

From a probabilistic point of view, much work has followed the seminal work of~\cite{FKV04} for the CSSP.~\cite{FKV04} introduced the idea of randomly sampling columns based on specific probability distributions.
~\cite{FKV04} use a simple probability distribution where each column of $\matA$ is sampled with probability proportional to its Euclidean norm. The approximation bound achieved, which holds only in expectation, is
$$
\FNormS{\matA - \matC \pinv{\matC}\matA}  \le \FNormS{\matA - \matA_k}  +(k/c)\FNormS{\matA}.
$$
\cite{DMM08} improved upon the accuracy of this result by using a distribution over the columns of $\matA$ where each column is sampled with probability proportional to its leverage score.
From a different perspective, ~\cite{DR10,GS12} presented some optimal algorithms using volume sampling.
~\cite{BDM11a} obtained faster optimal algorithms while~\cite{BW13} proposed optimal algorithms that run in input sparsity time.

Another line of research includes row-sampling algorithms for tall-and-skinny orthonormal matrices, which is relevant to our results: we essentially apply this kind of sampling to the rows of the matrix $\matV_k$ from the SVD of $\matA$. See Lemma~\ref{lem1} in the Section~\ref{sec:proofs} for a precise statement of our result. Similar results exist in~\cite{AB13}. We should also mention the work in~\cite{zouzias2012matrix}, which corresponds to a derandomization of the randomized sampling algorithm in~\cite{DMM08}. 

\section{Proofs}{\label{sec:proofs}}
Before we proceed, we setup some notation and definitions. For any two matrices $\matA$ and $\matB$ with appropriate dimensions, \math{\TNorm{\matA}\le\FNorm{\matA}\le\sqrt{\rank(\matA)}\TNorm{\matA}}, $\FNorm{\matA\matB} \leq \FNorm{\matA}\TNorm{\matB}$, and $ \FNorm{\matA\matB} \leq \TNorm{\matA} \FNorm{\matB}$.
$\XNorm{\matA}$ indicates that an expression holds for both $\xi = 2, \mathrm{F}$.
The thin (compact) Singular Value Decomposition (SVD) of a matrix $\matA \in \mathbb{R}^{m \times n}$ with $\rank(\matA) = \rho$ is:
\begin{eqnarray*}
\label{svdA} \matA
         = \underbrace{\left(\begin{array}{cc}
             \matU_{k} & \matU_{\rho-k}
          \end{array}
    \right)}_{\matU_{\matA} \in \R^{m \times \rho}}
    \underbrace{\left(\begin{array}{cc}
             \matSig_{k} & \bf{0}\\
             \bf{0} & \matSig_{\rho - k}
          \end{array}
    \right)}_{\matSig_\matA \in \R^{\rho \times \rho}}
    \underbrace{\left(\begin{array}{c}
             \matV_{k}\transp\\
             \matV_{\rho-k}\transp
          \end{array}
    \right)}_{\matV_\matA\transp \in \R^{\rho \times n}},
\end{eqnarray*}
with singular values \math{\sigma_1\left(\matA\right)\ge\ldots\sigma_k\left(\matA\right)\geq\sigma_{k+1}\left(\matA\right)\ge\ldots\ge\sigma_\rho\left(\matA\right) > 0}.
The matrices $\matU_\matA \in \R^{m \times \rho}$ and $\matV_{\matA} \in \R^{m \times (\rho)}$ contain the left and right singular vectors, respectively.
It is well-known that $\matA_k=\matU_k \matSig_k \matV_k\transp = \matU_k \matU_k\transp\matA = \matA\matV_k \matV_k\transp \in \R^{m \times n}$ minimizes \math{\XNorm{\matA - \matX}} over all
matrices \math{\matX \in \R^{m \times n}} of rank at most $k \le \rank(\matA)$.
The best rank-$k$ approximation to $\matA$ satisfies $\TNorm{\matA-\matA_k} = \sigma_{k+1}(\matA)$ and
$\FNormS{\matA-\matA_k} = \sum_{i=k+1}^{\rho}\sigma_{i}^2(\matA)$.
$\pinv{\matA}$
denotes the Moore-Penrose pseudo-inverse of $\matA $.
Let $\matB \in \R^{m \times n}$ ($m \le n$) and $\matA=\matB\matB\transp \in \R^{m \times m}$; then, for all $i=1, ...,m$, $\lambda_{i}\left(\matA\right) = \sigma_{i}^2\left(\matB\right)$ is the $i$-th eigenvalue of $\matA$.

\subsection{Proof of Theorem~\ref{thm1}}
To prove Theorem~\ref{thm1}, we will use the following result.
\begin{lemma}\label{lem:structural}
[Eqn.~3.2, Lemma 3.1 in~\cite{BDM11a}]
Consider $\matA = \matA \matZ \matZ\transp + \matE \in \R^{m \times n}$ as a
low-rank matrix factorization of $\matA$, with $\matZ \in \R^{n \times k},$
and $\matZ\transp\matZ=\matI_{k}$.
Let $\matS\in\R^{n\times c}$ ($c \ge k$) be any matrix such that 
$$rank(\matZ\transp \matS) =k.$$
Let $\matC = \matA \matS \in \R^{m \times c}$. Then, for $\xi=2,\mathrm{F}:$
$$
\XNormS{\matA - \matC \pinv{\matC} \matA} 
\le
\XNormS{\matA - \Pi_{\matC,k}^\xi(\matA)}
\le
\XNormS{\matE} \cdot \TNormS{\matS (\matZ\transp \matS)^\dagger}.
$$
Here, $\Pi_{\matC,k}^\xi(\matA) \in \mathbb{R}^{m \times n}$ is the best rank \math{k} 
approximation to \math{\matA} in the column space of \math{\matC} with respect to the $\xi$ norm. 
\end{lemma}
We will also use  the following novel lower bound on the smallest singular value of the matrix $\matV_k$,
after deterministic selection of its rows based on the largest leverage scores. 
\begin{lemma}\label{lem1}
Repeat the conditions of Theorem~\ref{thm1}.
Then, 
$$\sigma_k^2(\matV_k\transp \matS) > 1-\varepsilon .$$
\end{lemma}

\begin{proof}
We use the following perturbation result on the sum of eigenvalues of symmetric matrices.
\begin{lemma}\label{lem:weyl}[Theorem 2.8.1; part (i) in~\cite{brouwer2012spectra}]
Let $\matX$ and $\matY$ be symmetric matrices of order $k$ and, let $1 \le i,j \le n$
with $i+j \le k+1$.
Then,
$$
\lambda_i(\matX) \ge \lambda_{i+j-1}(\matX + \matY) - \lambda_j(\matY).
$$
\end{lemma}

Let $\matS \in \R^{n \times c}$ sample $c$ columns from $\matA$ with $c \ge k$.
Similarly, let $\hat\matS \in \R^{n \times (n-c)}$ sample the rest $n-c$ columns from $\matA$.
Hence,
$$
\matI_k = \matV_k\transp \matV_k = \matV_k\transp \matS \matS\transp \matV_k + \matV_k\transp \hat\matS \hat\matS\transp \matV_k.
$$
Let 
$$\matX = \matV_k\transp \matS \matS\transp \matV_k, \matY = \matV_k\transp \hat\matS \hat\matS\transp \matV_k, ~i=k, ~~\text{and}~~ j=1,$$ 
in Lemma~\ref{lem:weyl}. Notice that $i+j \le k+1,$
and $\lambda_{k}(\matX + \matY)=1$; hence:
\begin{align}
\lambda_k(\matV_k\transp \matS \matS\transp \matV_k) &\geq 1 - \lambda_1(\matV_k\transp \hat\matS \hat\matS\transp \matV_k ) \nonumber \\
												     &= 1 - \TNormS{\matV_k\transp \hat\matS} \nonumber \\
												     &\ge 1 -  \FNormS{\matV_k\transp \hat\matS} \nonumber \\
												     & > 1 - (k-\theta) \nonumber
\end{align}
Replacing $\theta = k - \varepsilon$ and
$\lambda_k(\matV_k\transp \matS \matS\transp \matV_k) = \sigma_k^2(\matV_k\transp\matS)$
concludes the proof.
\end{proof}

The proof of Theorem~\ref{thm1} is a straightforward combination of the Lemmas \ref{lem:structural} and \ref{lem1}.
First, by picking $\matZ = \matV_k$ in Lemma~\ref{lem:structural} we obtain:
\eqan{
\XNormS{\matA - \matC \pinv{\matC} \matA}
&\le& \XNormS{\matA-\matA_k} \cdot \TNormS{\matS (\matV_k\transp \matS)^\dagger} \\
&\le& \XNormS{\matA-\matA_k} \cdot \TNormS{\matS} \cdot \TNormS{(\matV_k\transp \matS)^\dagger}  \\
&=& \XNormS{\matA-\matA_k}                            \cdot \TNormS{ (\matV_k\transp \matS)^\dagger} \\
&=& \XNormS{\matA-\matA_k} / \sigma_k^2(\matV_k\transp \matS)
}
In the above, we used the facts that 
$$\matE = \matA - \matA \matV_k \matV_k\transp = \matA - \matA_k,$$ 
and the spectral norm of the sampling matrix $\matS$ equals one.
Also, we used that $\rank(\matV_k\transp\matS) = k$, which is implied from Lemma \ref{lem1}.
Next, via the bound in Lemma~\ref{lem1}:
$$ \XNormS{\matA - \matC \pinv{\matC} \matA}  < \XNormS{\matA-\matA_k} / (1-\varepsilon).$$

\subsection{Proof of Theorem~\ref{thm2}}
Let $\alpha_k = 1 + \eta$ for some $\eta > 0$. We assume that the leverage scores follow a power law decay such that:
$$\ell_i^{(k)} = \ell_1^{(k)} / i^{1 + \eta}.$$ 
According to the proposed algorithm, we select $c$ columns such that 
$$\sum_{i=1}^c \ell_i^{(k)} > \theta.$$ 
Here, we bound the number of columns $c$ required to achieve an $\varepsilon := k - \theta$ approximation in Theorem \ref{thm1}. To this end, we use the extreme case $\sum_{i=1}^c \ell_i^{(k)} = \theta$ which guarantees an $(1+\varepsilon)$-approximation. 

For our analysis, we use the following well-known result. 
\begin{proposition}\label{thm:integral}[Integral test for convergence]
Let $f(\cdot) \ge 0$ be  a function defined over the set of positive reals. Furthermore, assume that $f(\cdot)$ is monotone decreasing. Then, 
\begin{align}
\int_{j}^{J+1} f(i) dx ~\leq~ \sum_{i = j}^{J} f(i) ~\leq~ f(j) + \int_{j}^J f(x) dx, \nonumber
\end{align} over the interval $[j,\dots, J]$ for $j, J$ positive integers.
\end{proposition}

In our case, consider 
$$f(i) = \frac{1}{i^{1+\eta}}.$$ 
By definition of the leverage scores, we have:
$$k = \sum_{i=1}^n \ell_i^{(k)} = \ell_1^{(k)} \sum_{i=1}^n \frac{1}{i^{1+\eta}} \Longrightarrow \ell_1^{(k)} = \frac{k}{ \sum_{i=1}^n \frac{1}{i^{1+\eta}} }.$$
By construction, we collect $c$ leverage scores such that $k - \theta = \varepsilon$. This leads to: 
\begin{align}
k-\varepsilon &= \ell_1^{(k)}\cdot \sum_{i=1}^c  \frac{1}{i^{1+\eta}} =   \frac{k}{ \sum_{i=1}^n \frac{1}{i^{1+\eta}}} \cdot  \sum_{i=1}^c \frac{1}{i^{1+\eta}} \nonumber \\
&= k\left(\frac{\sum_{i=1}^n \frac{1}{i^{1+\eta}} - \sum_{i=c+1}^n \frac{1}{i^{1+\eta}}}{\sum_{i=1}^n \frac{1}{i^{1+\eta}}}\right) \nonumber \\
&= k\left(1- \frac{\sum_{i=c+1}^n \frac{1}{i^{1+\eta}}}{\sum_{i=1}^n \frac{1}{i^{1+\eta}}}\right) \Longrightarrow \nonumber \\
\varepsilon &= k \cdot \frac{\sum_{i=c+1}^n \frac{1}{i^{1+\eta}}}{\sum_{i=1}^n \frac{1}{i^{1+\eta}}}. \nonumber
\end{align}
To bound the quantity on the right hand side, we observe

\begin{align}
\frac{\sum_{i=c+1}^n \frac{1}{i^{1+\eta}}}{\sum_{i=1}^n \frac{1}{i^{1+\eta}}} &\le \frac{\frac{1}{(c+1)^{1+\eta}}+\int_{i=c+1}^n \frac{1}{x^{1+\eta}} dx }{\sum_{i=1}^n \frac{1}{i^{1+\eta}}} \nonumber \\
& = \frac{\frac{1}{(c+1)^{1+\eta}}+\left[-\frac{1}{x^{1+\eta}}\right]_{i=c+1}^n }{\sum_{i=1}^n \frac{1}{i^{1+\eta}}} \nonumber \\
& = \frac{\frac{1}{(c+1)^{1+\eta}}+\frac{1}{\eta}\left(\frac{1}{(c+1)^{\eta}}-\frac{1}{n^{\eta}}\right)}{\sum_{i=1}^n \frac{1}{i^{1+\eta}}} \nonumber \\
& \le \frac{\frac{1}{(c+1)^{1+\eta}}+\frac{1}{\eta(c+1)^{\eta}}}{\sum_{i=1}^n \frac{1}{i^{1+\eta}}} \nonumber \\
& = \frac{\frac{1}{(c+1)^{1+\eta}}+\frac{1}{\eta(c+1)^{\eta}}}{1+\sum_{i=2}^n \frac{1}{i^{1+\eta}}} \nonumber \\
&< \frac{1}{(c+1)\cdot (c+1)^{\eta}}+\frac{1}{\eta(c+1)^{\eta}} \nonumber \\
& \le \max\left\{ \frac{2}{ (c+1)^{1+\eta}},\;\;\;\frac{2}{\eta\cdot (c+1)^{\eta}} \right\}  \nonumber
\end{align} where the first inequality is due to the right hand side of the integral test and the third inequality is due to 
$$1+\sum_{i=2}^n \frac{1}{i^{1+\eta}} > 1.$$
Hence, we may conclude: 
$$
\varepsilon < k\cdot  \max\left\{ \frac{2}{ (c+1)^{1+\eta}},\;\;\;\frac{2}{\eta\cdot (c+1)^{\eta}} \right\}.
$$
The above lead to the following two cases: if 
$$\varepsilon < \frac{2k}{ (c+1)^{1+\eta}},$$ 
we have: 
$$c < \left(\frac{2\cdot k}{\varepsilon}\right)^{\frac{1}{1+\eta}}-1,$$ 
whereas in the case where 
$$\varepsilon < \frac{2\cdot k}{\eta\cdot (c+1)^{\eta}},$$ 
we get 
$$c < \left(\frac{2\cdot k}{\eta \cdot \varepsilon}\right)^{\frac{1}{\eta}}-1.$$ 

\section{The key role of $\theta$}\label{sec:theta}
In the proof of Theorem~\ref{thm1}, we require that 
$$ \sigma_k^2(\matV_k\transp \matS) > 1 - (k - \theta) := 1 - \varepsilon.$$ 
For this condition to hold, the sampling matrix $\matS$ should preserve the rank of $\matV_k\transp$ in $\matV_k\transp \matS,$
i.e., choose $\theta $ such that $\rank(\matV_k\transp \matS) = k$.

Failing to preserve the rank of $\matV_k\transp$ has immediate implications for the CSSP.
To highlight this, let $\matA \in \mathbb{R}^{m \times n}$ of rank $k < \min\{m,n\}$ with SVD $\matA = \matU_k \matSig_k \matV_k\transp$.
Further, assume that the $k$th singular value of $\matA$ is arbitrary large, i.e.,
$\sigma_k(\matA) \rightarrow \infty$. Also, let $\rank(\matV\transp \matS) = \gamma < k$ and $\matC = \matA \matS$.
Then,
\begin{align*}
&\|\matA - \matC \pinv{\matC} \matA\|_{\xi}  \nonumber \\ &=\XNorm{ \matU_k \matSig_k \matV_k\transp - \matU_k \matSig_k \matV_k\transp\matS (\matU_k \matSig_k \matV_k\transp\matS)^\dagger \matU_k \matSig_k \matV_k\transp} \\
&= \XNorm{  \matSig_k  - \matSig_k \matV_k\transp\matS (\matU_k \matSig_k \matV_k\transp\matS)^\dagger \matU_k \matSig_k }\\
&=  \XNorm{  \matSig_k  - \matSig_k \matV_k\transp\matS (\matSig_k\matV_k\transp\matS)^\dagger (\matU_k )^\dagger \matU_k \matSig_k }\\
&=  \XNorm{  \matSig_k  - \matSig_k \matV_k\transp\matS (\matSig_k\matV_k\transp\matS)^\dagger \matSig_k }\\
&=  \XNorm{  \matSig_k  - \matU_{\matX} \matU_{\matX}\transp \matSig_k }\\
& \ge \sigma_k(\matA)
\end{align*}
The second equality is due to the fact that both spectral and Frobenius norms are invariant to unitary transformations.
In the third equality, we used the fact that $(\matW \matZ )^\dagger = \matZ^\dagger \matW^\dagger$
if $\matW\transp\matW$ is the identity matrix. Then, set $\matX = \matSig_k \matV_k\transp\matS \in \R^{k \times c}$
where $\rank(\matX) = \gamma$. Using this notation, let $ \matU_{\matX} \in \R^{m \times \gamma}$ be any orthonormal basis for
$span(\matX)$. Observe $\matU_{\matX} \matU_{\matX}\transp = \matX \matX^\dagger.$
The last inequality is due to $\matU_{\matX} \matU_{\matX}\transp$ being an $m \times m$ diagonal
matrix with $\gamma$ ones along its main diagonal and the rest zeros.
Thus, we may conclude that for this $\matA$: 
$$\XNorm{\matA - \matC \pinv{\matC}\matA} \ge \sigma_k(\matA) \rightarrow \infty.$$

\section{Extensions to main algorithm}\label{sec:extensions}

Algorithm~\ref{alg1} requires $O( \min\{m,n\} m n)$ arithmetic operations since, in the first step of the algorithm,
it computes the top k right singular vectors of $\matA$ through the SVD. In this section,
we describe how to improve the running time complexity of the algorithm while maintaining about
the same approximation guarantees. The main idea is to replace
the top $k$ right singular vectors of $\matA$ with some orthonormal vectors that ``approximate''
the top $k$ right singular vectors in a sense that we make precise below.  
Boutsidis et al. introduced this idea to improve the running time complexity of column subset selection algorithms
in~\cite{BDM11a}. 

\subsection{Frobenius norm}
We start with a result which is a slight extension of a result appeared in~\cite{GP13}. 
Lemma~\ref{lem:approxSVDdet} below appeared in~\cite{BW13} but we include the proof
for completeness. For the description of the algorithm we refer to~\cite{Lib13,GP13}.
\begin{lemma}[Theorem 3.1 in~\cite{GP13}]\label{lem:approxSVDdet}
Given \math{\matA\in\R^{m\times n}} of rank $\rho$, a target rank $1 \leq k < \rho$, and $0 < \varepsilon \le 1$, there exists a deterministic algorithm that  computes
$\matZ \in \R^{n \times k}$ with $\matZ\transp\matZ = \matI_k$  and
$$\FNormS{\matA - \matA \matZ \matZ\transp} \leq \left(1+{\varepsilon}\right)\FNormS{\matA - \matA_k}.$$
The proposed algorithm runs in $O\left(m n k^2 \varepsilon^{-2}\right)$ time. 
\end{lemma}
\begin{proof}
Theorem 4.1 in~\cite{GP13} describes an algorithm that given $\matA, k$ and $\varepsilon$ constructs $\matQ_k \in \R^{k \times n}$
such that $$\FNormS{\matA - \matA\matQ_k\transp\matQ_k} \leq \left(1+{\varepsilon}\right)\FNormS{\matA - \matA_k}.$$ To obtain the desired
factorization, we just need an additional step to ortho-normalize the columns of $\matQ_k\transp,$ which takes $O(nk^2)$ time. So, assume
that $\matQ_k\transp = \matZ \matR$ is a QR factorization of $\matQ_k\transp$ with $\matZ \in \R^{n \times k}$ and $\matR \in \R^{k \times k}$.
Then,
\begin{align*}
\FNormS{\matA - \matA \matZ \matZ\transp} &\le \FNormS{\matA - \matA \matQ_k\transp  (\matR\transp) \matZ} \nonumber \\
&= \FNormS{\matA - \matA \matQ_k\transp \matR\transp (\matR\transp)^{-1} \matQ_k} \nonumber  \\
&= \FNormS{\matA - \matA \matQ_k\transp  \matQ_k} \nonumber \\
&\le \left(1+{\varepsilon}\right)\FNormS{\matA - \matA_k}.
\end{align*}
\end{proof}
In words, the lemma describes a method that constructs a rank $k$ matrix $\matA \matZ \matZ\transp$ that
is as good as the rank $k$ matrix $\matA_k = \matA \matV_k \matV_k\transp$ from the SVD of $\matA$.
Hence, in that ``low-rank matrix approximation sense'' $\matZ$ can replace $\matV_k$ in our column subset selection algorithm.

Now, consider an algorithm as in Algorithm~\ref{alg1} where in the first step,
instead of $\matV_k,$ we compute
$\matZ$ as it was described in Lemma~\ref{lem:approxSVDdet}. This modified algorithm requires
$O\left(m n k^2 \varepsilon^{-2}\right)$ arithmetic operations.  For this deterministic 
algorithm we have the following theorem. 
\begin{theorem}\label{thm1_ext1}
Let 
$\theta = k- \varepsilon,$ for some $\varepsilon \in (0,1)$, and 
let $\matS$ be the $n \times c$ output  sampling matrix of the modified Algorithm~\ref{alg1} described above.
Then, for $\matC=\matA \matS$ we have 
$$ \FNormS{\matA - \matC \pinv{\matC}\matA} <  \left( 1 + \varepsilon \right) \cdot \left(1-\varepsilon\right)^{-1} \cdot \FNormS{\matA - \matA_k}. $$
\end{theorem}
\begin{proof}
Let $\matZ$ be constructed as in Lemma~\ref{lem:approxSVDdet}. 
Using this $\matZ$ and $\xi = \mathrm{F}$ in Lemma~\ref{lem:structural} we obtain:
\eqan{
\FNormS{\matA - \matC \pinv{\matC} \matA}
&\le& \FNormS{\matA-\matA \matZ \matZ\transp} \cdot \TNormS{\matS (\matZ\transp \matS)^\dagger} \\
&\le& \FNormS{\matA-\matA \matZ \matZ\transp} \cdot \TNormS{\matS} \cdot \TNormS{(\matZ\transp \matS)^\dagger}  \\
&=& \FNormS{\matA-\matA \matZ \matZ\transp}                            \cdot \TNormS{ (\matZ\transp \matS)^\dagger} \\
&=& \FNormS{\matA-\matA \matZ \matZ\transp} / \sigma_k^2(\matZ\transp \matS)
}
In the above, we used the facts that $\matE = \matA -\matA \matZ \matZ\transp$ and the spectral norm of the sampling matrix $\matS$ equals one.
Also, we used that $\rank(\matZ \transp\matS) = k$, which is implied from Lemma \ref{lem1}.
Next, via the bound in Lemma~\ref{lem1} on $\matZ$~\footnote{
It is easy to see that Lemma~\ref{lem1} holds for any orthonormal matrix $\matV_k$
and it is not neccesary that $\matV_k$ contains the singular vectors of  matrix $\matA$.
}:
$$ \TNormS{\matA - \matC \pinv{\matC} \matA}  < \TNormS{\matA-\matA \matZ \matZ\transp} / (1-\varepsilon).$$
Finally, using $\FNormS{\matA - \matA \matZ \matZ\transp} \leq \left(1+{\varepsilon}\right)\FNormS{\matA - \matA_k}$
according to  Lemma~\ref{lem:approxSVDdet} concludes the proof. 
\end{proof}

\subsection{Spectral norm}
To achieve a similar running time improvement for the spectral norm bound of Theorem~\ref{thm1},
we need an analogous result as in Lemma~\ref{lem:approxSVDdet}, but for the spectral norm. We
are not aware of any such deterministic algorithm. Hence, we quote Lemma 11 from~\cite{BDM11a},
which provides a randomized algorithm. 
\begin{lemma}[Randomized fast spectral norm SVD]
\label{tropp1}
Given \math{\matA\in\R^{m\times n}} of rank $\rho$, a target rank $2\leq k < \rho$, and
$0 < \varepsilon < 1$,
there exists an algorithm that
computes a factorization  $\matA = \matB \matZ\transp + \matE$, with $\matB = \matA \matZ$, $\matZ\transp\matZ = \matI_k$
such that
$$
\Expect{\TNorm{\matE}} \leq \left(\sqrt{2}+\varepsilon\right) \TNorm{\matA - \matA_k}.
$$
This algorithm runs in
$O\left(mnk\varepsilon^{-1}
\log\left( k^{-1}\min\{m,n\}\right)\right)$ time.
\end{lemma}
In words, the lemma describes a method that constructs a rank $k$ matrix $\matA \matZ \matZ\transp$ that
is as good as the rank $k$ matrix $\matA_k = \matA \matV_k \matV_k\transp$ from the SVD of $\matA$.
Hence, in that ``low-rank matrix approximation sense'' $\matZ$ can replace $\matV_k$. The difference between
Lemma~\ref{tropp1} and Lemma~\ref{lem:approxSVDdet} is that the matrix $\matA \matZ \matZ\transp$
approximates $\matA_k$ with respect to a different norm.

Now consider an algorithm as in Algorithm~\ref{alg1} where in the first step we compute
$\matZ$ as it was described in Lemma~\ref{tropp1}. This algorithm takes 
$O\left(mnk\varepsilon^{-1}
\log\left( k^{-1}\min\{m,n\}\right)\right)$ time.  For this randomized
algorithm we have the following theorem. 
\begin{theorem}\label{thm1_ext2}
Let 
$\theta = k- \varepsilon,$ for some $\varepsilon \in (0,1)$, and 
let $\matS$ be the $n \times c$ output  sampling matrix of the modified Algorithm~\ref{alg1} described above.
Then, for $\matC=\matA \matS$ we have 
$$ \Expect{\TNorm{\matA - \matC \pinv{\matC}\matA}} <  \left( \sqrt{2} + \varepsilon \right) \cdot \sqrt{ \left(1-\varepsilon\right)^{-1}} \cdot \TNorm{\matA - \matA_k}. $$
\end{theorem}
\begin{proof}
Let $\matZ$ be constructed as in Lemma~\ref{tropp1}. 
Using this $\matZ$ and $\xi = 2$ in Lemma~\ref{lem:structural} we obtain:
\eqan{
\TNormS{\matA - \matC \pinv{\matC} \matA}
&\le& \TNormS{\matA-\matA \matZ \matZ\transp} \cdot \TNormS{\matS (\matZ\transp \matS)^\dagger} \\
&\le& \TNormS{\matA-\matA \matZ \matZ\transp} \cdot \TNormS{\matS} \cdot \TNormS{(\matZ\transp \matS)^\dagger}  \\
&=& \TNormS{\matA-\matA \matZ \matZ\transp}                            \cdot \TNormS{ (\matZ\transp \matS)^\dagger} \\
&=& \TNormS{\matA-\matA \matZ \matZ\transp} / \sigma_k^2(\matZ\transp \matS)
}
In the above, we used the facts that $\matE = \matA -\matA \matZ \matZ\transp$ and the spectral norm of the sampling matrix $\matS$ equals one.
Also, we used that $\rank(\matZ \transp\matS) = k$, which is implied from Lemma \ref{lem1}.
Next, via the bound in Lemma~\ref{lem1} on $\matZ$:
$$ \TNormS{\matA - \matC \pinv{\matC} \matA}  < \TNormS{\matA-\matA \matZ \matZ\transp} / (1-\varepsilon).$$
Taking square root on both sides of this relation we obtain:
$$ \TNorm{\matA - \matC \pinv{\matC} \matA}  < \TNorm{\matA-\matA \matZ \matZ\transp} \sqrt{ (1-\varepsilon)^{-1} }.$$
Taking expectations with respect to the randomness of $\matZ$ yields,
$$ \Expect{\TNorm{\matA - \matC \pinv{\matC} \matA}}  < \Expect{\TNorm{\matA-\matA \matZ \matZ\transp}} \sqrt{ (1-\varepsilon)^{-1} }.$$
Finally, using $\Expect{\TNorm{\matA-\matA \matZ \matZ\transp}} \le (\sqrt{2} + \varepsilon) \TNorm{\matA - \matA_k}$ - from Lemma~\ref{tropp1} -
concludes the proof. 
\end{proof}

We also mention that it is now straightforward to prove an analog of Theorem~\ref{thm2}
for the algorithms we analyze in Theorems~\ref{thm1_ext1} and~\ref{thm1_ext2}. One should
replace the assumption of the power law decay of the leverage scores with an assumption
of the power law decay of the row norms square of the matrix $\matZ$. Whether the row norms
of those matrices $\matZ$ follow a power law decay is an interesting open question which will
be worthy to investigate in more detail. 

\subsection{Approximations of rank $k$}
Theorems~\ref{thm1},~\ref{thm2},~\ref{thm1_ext1}, and~\ref{thm1_ext2} provide bounds for low rank approximations
of the form $\matC \matC^\dagger \matA \in \R^{m \times n},$ where $\matC$ contains $c \ge k$ columns of $\matA$.  The matrix 
$\matC \matC^\dagger \matA$ could potentially have rank larger than $k$, indeed it can be as large as $c$. 
In this section, we describe how to construct factorizations that have rank $k$ and are as accurate as those in Theorems~\ref{thm1},~\ref{thm2},~\ref{thm1_ext1}, and~\ref{thm1_ext2}. Constructing a rank $k$ instead of a rank $c$ column-based low-rank matrix factorization is a harder problem
and might be desirable in certain applications (see, for example, Section 4 in~\cite{DRVW06} where the authors apply rank $k$ column-based low-rank matrix factorizations to solve the projective clustering problem). 

Let $\matA \in \mathbb{R}^{m \times n}$, let $k < n$ be an integer, and let $\matC \in \mathbb{R}^{m \times c}$ with $c \ge k$. Let $\Pi_{\matC,k}^\xi(\matA) \in \mathbb{R}^{m \times n}$ be the best rank \math{k} approximation to \math{\matA} in the column space of \math{\matC} with respect to the $\xi$ norm. 
Hence, we can write $\Pi_{\matC,k}^\xi(\matA) = \matC\matX^\xi$, where
$$
\matX^\xi = \argmin_{\matPsi \in {\R}^{c \times n}:\rank(\matPsi)\leq k}\XNormS{\matA-
\matC\matPsi}.
$$
In order to compute (or approximate) $\Pi_{\matC,k}^{\xi}(\matA)$ given $\matA$,
$\matC$, and $k$, we will use the following algorithm:
\begin{center}
\begin{algorithmic}[1]
\STATE Ortho-normalize the columns of $\matC \in \R^{m \times c}$ in $O(m c^2)$ time to construct the matrix $\matQ \in \R^{m \times c}$.
\STATE Compute
 $(\matQ\transp \matA)_k \in \R^{c \times n}$ via the SVD
in \math{O(mnc+ nc^2)}; $(\matQ\transp \matA)_k$ has rank $k$ and denotes the best rank-$k$ approximation of
\math{\matQ\transp\matA}.
\STATE Return 
$\matQ(\matQ\transp \matA)_k \in \mathbb{R}^{m \times n}$ in $O(mnk)$ time.
\end{algorithmic}
\end{center}
\medskip
Clearly, $\matQ(\matQ\transp \matA)_k$ is a rank $k$ matrix that lies in the column span of $\matC$. Note that though  $\Pi_{\matC,k}^{\xi}(\matA)$ can depend on \math{\xi}, our algorithm computes the same matrix, independent of \math{\xi}. The next lemma was proved in~\cite{BDM11a}. 
\begin{lemma}\label{lem:bestF}
Given $\matA \in {\R}^{m \times n}$, $\matC\in\R^{m\times c}$ and an integer $k$,  the matrix $\matQ(\matQ\transp \matA)_k \in \mathbb{R}^{m \times n}$ described above (where \math{\matQ} is an orthonormal basis for the columns of \math{\matC}) can be computed in $O\left(mnc + (m+n)c^2\right)$ time and satisfies:
\begin{eqnarray*}
\FNormS{\matA-\matQ(\matQ\transp \matA)_k} &=& \FNormS{\matA-\Pi_{\matC,k}^{\mathrm{F}}(\matA)},\\[3pt]
\TNormS{\matA-\matQ(\matQ\transp \matA)_k} &\leq& 2\TNormS{\matA-\Pi_{\matC,k}^{2}(\matA)}.
\end{eqnarray*}
\end{lemma}

Finally, observe that Lemma~\ref{lem:structural} indeed provides an upper bound for the residual
error $\XNormS{\matA - \Pi_{\matC,k}^{\xi}(\matA)}$. Hence, all bounds in 
Theorems~\ref{thm1},~\ref{thm2},~\ref{thm1_ext1}, and~\ref{thm1_ext2} 
hold for the error $\XNormS{\matA - \Pi_{\matC,k}^{\xi}(\matA)}$ as well,
and by Lemma~\ref{lem:bestF} one can provide bounds for the errors
$\FNormS{\matA-\matQ(\matQ\transp \matA)_k}$ and 
$\TNormS{\matA-\matQ(\matQ\transp \matA)_k}$.

\section{Concluding Remarks}
We provided a rigorous theoretical analysis of an old and popular \emph{deterministic} feature selection algorithm from the statistics literature~\cite{Jol72}. Although randomized algorithms are often easier to analyze, we believe that deterministic algorithms are simpler to implement and explain, hence more attractive to practitioners and data analysts. 

One interesting path for future research is understanding the connection of this work with the so-called ``spectral graph sparsification problem''~\cite{SS08}. In that case, edge selection in a graph is implemented via randomized leverage scores sampling from an appropriate matrix (see Theorem 1 in~\cite{SS08}). Note that in the context of graph sparsification, leverage scores correspond to the so-called ``effective resistances'' of the graph. Can deterministic effective resistances sampling be rigorously analyzed? What graphs have effective resistances following a power law distribution?

{
\bibliographystyle{abbrv}

}
\end{document}